\documentclass[journal,twoside]{asmb}
\newenvironment{proof}{\paragraph{Proof:}}{\hfill$\square$}
\begin{document}
\title{Dual Threats in RIS-Aided RF-UOWC Mixed Networks: Secrecy Performance Analysis under Simultaneous RF and UOWC Eavesdropping}

\author{
\IEEEauthorblockN{Md. Abdur Rakib, Md. Ibrahim,
\textit{Graduate Student Member, IEEE}, A. S. M. Badrudduza, \textit{ Member, IEEE}, Imran Shafique Ansari, \textit{Senior Member, IEEE}}}

\twocolumn[
\begin{@twocolumnfalse}
\maketitle
\begin{abstract}
\section*{Abstract}

In the dynamic realm of $6$G technology, emphasizing security is essential, particularly for optimizing high-performance communication. A notable strategy involves the use of reconfigurable intelligent surfaces (RISs), an emerging and cost-efficient technology aimed at fortifying incoming signals, broadening coverage, and ultimately improving the overall performance of systems. In this paper, we introduce a comprehensive framework to analyze the secrecy performance of an RIS-assisted mixed radio frequency (RF) - underwater optical wireless communication (UOWC) network. Here, all the RF links undergo $\alpha$-$\mu$ fading distribution, whereas the UOWC links experience a mixture of Exponential Generalized Gamma distribution. Specifically, we examine three potential eavesdropping situations: 1) eavesdropping on the RF link, 2) eavesdropping on the UOWC link, and 3) a simultaneous eavesdropping attack affecting both RF and UOWC links. To achieve this, we derive novel mathematical expressions such as average secrecy capacity, secrecy outage probability, strictly positive secrecy capacity, and effective secrecy throughput in closed form. Using these derived expressions, we carry out an investigation to assess the influences of fading parameters, pointing errors, receiver detection technique, underwater turbulence severity, and water salinity on the system. Furthermore, our study investigates the significance of RIS in improving secrecy performance due to the proposed model. To provide deeper insights, we also perform asymptotic analysis for the high signal-to-noise region. Finally, to verify our analytical results, we conduct Monte Carlo simulation using a computer-based technique.
\end{abstract}

\begin{IEEEkeywords}
\section*{Keywords} 

Reconfigurable intelligent surface (RIS), $\alpha-\mu$ channel, underwater turbulence, and secrecy outage probability.
\end{IEEEkeywords}
\end{@twocolumnfalse}
]
\section{INTRODUCTION}
\subsection{Background}
With the emergence of the $6$G wireless network, reconfigurable intelligent surfaces (RISs) are increasingly recognized as a pivotal technology, particularly in the realms of radio frequency (RF) and underwater optical communication (UOWC) links. These surfaces are composed of an array of cost-effective passive elements adept at modulating the amplitude, frequency, phase, and polarization of incoming signals in real-time, thereby establishing a programmable propagation environment\cite{art14}. The use of RIS techniques shows significant potential in multiple areas, particularly for enhancing the signal-to-noise ratio (SNR), increasing signal coverage, and improving beamforming in multi-user channels\cite{art16}.
In modern times, UOWC has attracted considerable attention due to its incorporation of various communication techniques with RF, establishing itself as a prospective advancement for the next generation of cellular networks\cite{art17}.
Through the integration of RIS in both RF and optical transmissions, the persistent challenge of signal blockage, inherent in wireless communication, can be effectively mitigated. 
\subsection{Literature Study}
As wireless networks rapidly evolve, researchers increasingly highlight the importance of security concerns. In particular, $6$G Internet of Things (IoT) networks face significant security challenges due to the broadcast nature of wireless channels and the transmission of large volumes of sensitive data. Consequently, physical layer security (PLS) has emerged as a leading and highly effective method for ensuring secure communication \cite{art42, art43, art44, lo6}. Radio frequency (RF) communication offers advantages like longer range and penetration through obstacles, whereas underwater optical wireless communication (UOWC) provides higher data rates and lower latency. Recently, dual-hop RF-UOWC networks have been deployed to leverage the strengths of each technology while mitigating their respective limitations \cite{art30, art31, art32, lo7}. The secrecy performance of mixed RF-UOWC was analyzed in \cite{art5}, considering potential eavesdroppers on the RF link. The authors derived closed-form expressions for various secrecy metrics, such as secrecy outage probability (SOP), average secrecy capacity (ASC), and the probability of strictly positive secrecy capacity (SPSC). Another study in \cite{art41} analyzed a similar model to enhance confidentiality, where eavesdroppers attempt to intercept the RF link without being aware of the antenna selection scheme. Additionally, research in \cite{new4} conducted secrecy outage analysis for decode-and-forward (DF) relay-based networks using an $\alpha-\mu$ distribution, which is a versatile channel model capable of simulating a wider range of realistic physical scenarios. Recently, the impact of fading severity and various turbulence scenarios was examined in \cite{art16}.

Due to their significant impact on modern wireless communication networks, RISs and their typical applications have garnered considerable interest among researchers. The studies in \cite{art18,art19,art20,art21,new5,new6,new7,lo2,lo3} focus on single-hop RIS-assisted systems. In \cite{art18}, researchers conducted a comprehensive coverage analysis of RIS-aided communication systems, demonstrating that RIS deployment can expand coverage, enhance SNR, and reduce transmission delays. The performance of RIS-assisted single-input single-output networks was assessed in \cite{art19}. Further examination of RIS capabilities was performed in \cite{art20}, where the performance of wireless communication systems with RIS and two different phase shifting designs was analyzed, particularly in Nakagami-m fading channels. The study in \cite{art21} highlighted the improvement in secrecy performance achievable through RIS implementation in wireless communication networks, underscoring its role in enhancing security. Another study, \cite{new5}, investigated the secrecy performance of RIS-aided systems over Rayleigh fading channels, revealing the significant impact of the number of reflecting elements on system secrecy performance by analyzing the exact closed-form expressions for SOP, intercept probability (IP), and the probability of non-zero secrecy capacity (PNZSC).

In recent years, numerous studies have explored the effectiveness of dual-hop RIS-aided systems \cite{art16,ahmed2023enhancing,new9,new10,new11,art22,art23,art26,art27,art28,art29,10413214,art24,lo4,lo5}. One such study in \cite{art16}, analyzed the security aspects of RF-UWOC networks, demonstrating how integrating RIS into the RF channel can enhance secure data transmission, especially in scenarios with potential eavesdroppers. In \cite{ahmed2023enhancing}, researchers investigated the physical layer security of RIS-assisted mixed RF-FSO frameworks across three different eavesdropping scenarios. A study presented in \cite{new9} addressed RIS-based jamming techniques to prevent eavesdropping in mixed RF-FSO systems, considering Nakagami-$m$ fading for RF links and Málaga ($\mathcal{M}$) turbulence for FSO links. They evaluated secrecy metrics such as secrecy outage probability and average secrecy capacity, thoroughly examining various system parameters to assess their impact on security performance. The authors in \cite{new10} analyzed the secrecy performance of RIS-assisted systems under simultaneous eavesdropping in both RF and FSO links, using metrics including ASC, SOP, probability of SPSC, effective secrecy throughput (EST), and IP.
Furthermore, the existing literature provides limited insights into the performance of RIS-assisted dual-hop systems. In \cite{art26}, the performance of an RIS-equipped source in a mixed RF-FSO relay network was evaluated, showing that the RIS-equipped source outperformed the RIS-aided scenario. The impact of co-channel interference (CCI) on an RIS-aided dual-hop FSO-RF system was examined in \cite{art27}. Authors in \cite{art29} proposed a RIS-based mixed FSO-RF system, taking into account pointing errors and atmospheric turbulence, and found that RIS could significantly enhance system performance based on their diversity order analysis. Another study, \cite{art24}, introduced and analyzed a RIS-empowered dual-hop mixed RF-UOWC model using various detection techniques.

\subsection{Motivation and Contributions}
In light of the aforementioned studies, the prevailing literature commonly highlights a specific emphasis on RF and UOW systems, acknowledging their significant impact on the evolving landscape of 6G wireless communication technology. While RF channels in wireless communication offer advantages such as reliable signal penetration, extended communication ranges, and established infrastructure, facilitating widespread adoption and connectivity, they are susceptible to interference, limited bandwidth, and potential signal degradation in crowded or noisy environments, negatively impacting overall communication reliability and performance\cite{new1}. In contrast, UOWC holds the potential for achieving high data rates, bandwidth, and transmission security, along with substantial transmission distances\cite{art31}. However, point-to-point UOWC is often susceptible to the effects of underwater turbulence (UWT) \cite{new3}. Consequently, a mixed dual-hop relaying system emerges as a viable solution to address the inherent limitations of single-hop transmission, leveraging the advantages of both types of links. While the current literature predominantly explores mixed RF-UOWC systems, there is a notable scarcity of research on their potential to preserve secrecy. Furthermore, existing research lacks a comprehensive investigation into the security aspects of RIS-assisted RF-UOW systems, particularly in scenarios involving simultaneous eavesdropping. This gap is especially pronounced in the context of physical layer security (PLS), where RIS is utilized in both communication links. The analysis of secrecy performance in RIS-assisted mixed RF-UOWC systems under simultaneous eavesdropping scenarios remains largely unexplored.
In this paper, we introduce a novel framework for a RIS-aided mixed RF-UOWC system, where the RF link experiences $\alpha-\mu$ fading and the UOWC link follows an mEGG distribution. The $\alpha-\mu$ distributions can account for diverse fading channels, and the mEGG model incorporates various real-world physical transformations, including underwater turbulence, air bubbles, temperature gradients, and water salinity \cite{zedini2019unified,art5}. A strategically positioned relay node between the source and the user captures signals from the source, converts them into optical signals, and then transmits them to the intended user through the UOWC link. The key contributions of this study are outlined as follows:
\begin{itemize}
\item Firstly, we formulate the closed-form expressions for the cumulative distribution function (CDF) and probability distribution function (PDF) of the dual-hop RF-UOWC network, utilizing the CDF and PDF of individual links. Furthermore, we demonstrate the analytical expressions for ASC, SOP, SPSC, and EST in closed form. Notably, when compared to prior research, these formulations exhibit novelty as the proposed model stands out uniquely from the current body of literature.

\item Utilizing the obtained analytical expressions, we produce numerical outcomes along with accompanying figures. To validate the accuracy of these analytical results, we incorporate Monte-Carlo simulation, thereby enhancing the robustness of the analysis. Additionally, to gain deeper physical insights from the analytical results, we derive asymptotic expressions for the performance metrics in the high SNR regime.
\item To conduct a more practical and realistic analysis, we consider fundamental impairments and influential factors in both the RF and UOW links, influencing the secrecy performance of the proposed model. These factors encompass the effects of average SNR, the quantity of RIS reflecting elements, fading parameters of the RF link, environmental variables such as air bubbles, temperature gradients, water salinity, and the pointing error of the UOW link.
\end{itemize}
\subsection{Organization}
The organization of this paper is as follows: Section \ref{model} details the system and channel models. In Section \ref{spa}, we derive closed-form expressions for ASC, SOP, SPSC, and EST. Section \ref{NR} presents the analytical and simulation results. Finally, Section \ref{con} provides concluding remarks.


\section{SYSTEM MODEL AND PROBLEM FORMULATION}
\label{model}
\begin{figure*}[!ht]
   \vspace{-5mm}
       \centerline{\includegraphics[width=0.7\textwidth,angle =0]{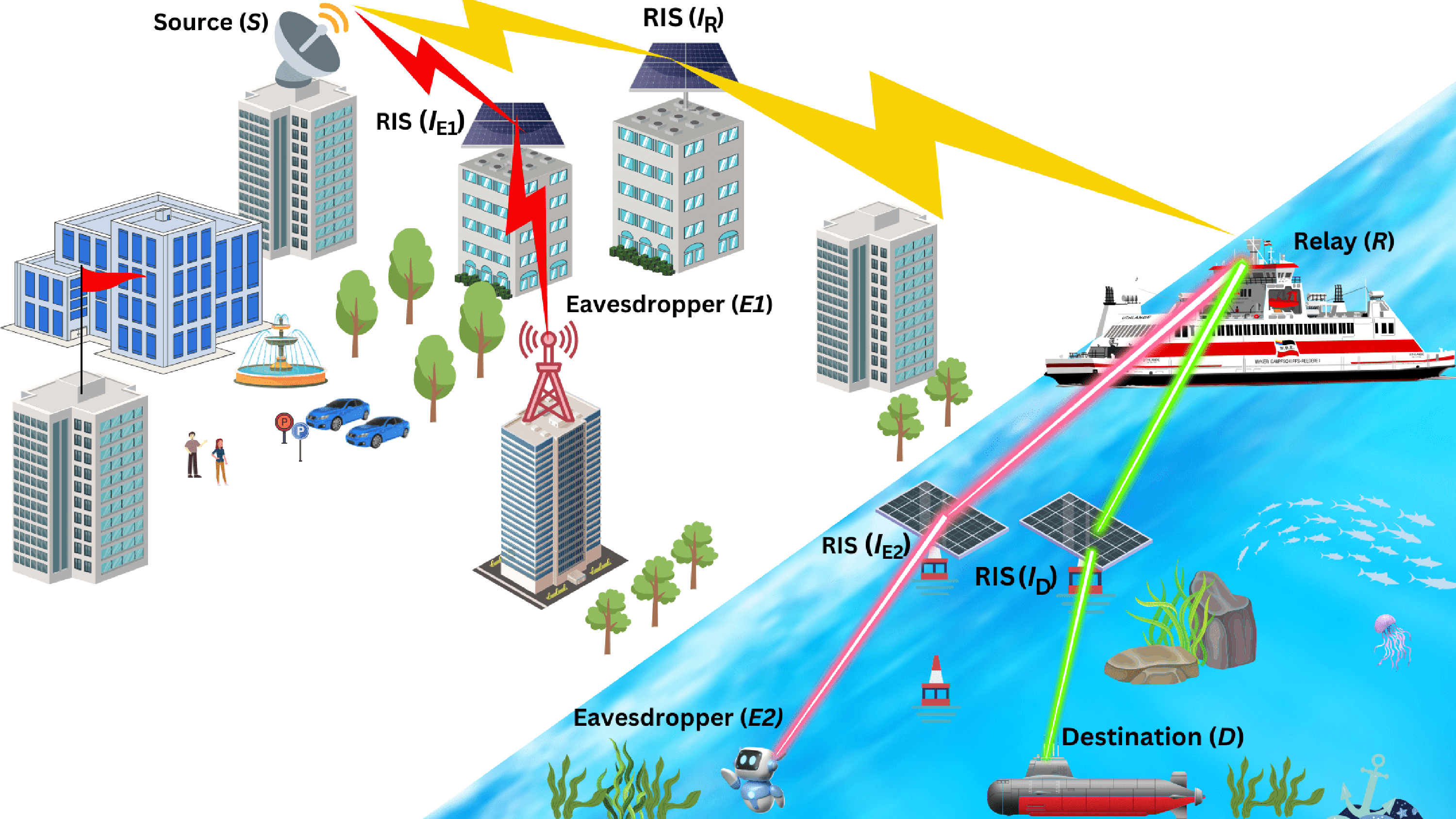}}
         \vspace{0mm}
    \caption{The system model incorporates a source ($\mathbf{S}$), a relay ($\mathcal{R}$), a destination user ($\mathcal{D}$), and two eavesdroppers ($\mathcal{E}1$ and $\mathcal{E}2$), along with four RISs ($I_{\mathcal{R}}$, $I_{\mathcal{D}}$, $I_{\mathcal{E}1}$, and $I_{\mathcal{E}2}$).}
       \label{mod}
   \end{figure*}
We consider a dual-hop RIS-assisted RF-UOWC mixed communication model as shown in Fig.\ref{mod}. In the first hop, we utilize a RIS-RF system, while in the second hop, we employed a RIS-UOWC system. Due to the presence of physical obstacle, it is not feasible for source, $\mathcal{S}$ to create direct communication link with relay, $\mathcal{R}$. Therefore, $\mathcal{S}$ transmits signal to $\mathcal{R}$ by employing a RIS, $\mathcal{I}_{\mathcal{R}}$ placed on the top of a building. Likewise, another RIS, $\mathcal{I}_{\mathcal{D}}$ facilitates the optical communication between $\mathcal{R}$ and destination, $\mathcal{D}$. The eavesdroppers (i.e., $\mathcal{E}1$ and $\mathcal{E}2$) try to intercept the confidential information being transmitted from $\mathcal{S}$ to $\mathcal{D}$ via two different RISs (i.e., $\mathcal{I}_{\mathcal{E}1}$ and $\mathcal{I}_{\mathcal{E}2}$). Based on the eavesdroppers location, three distinct scenarios are considered in this paper. 
\begin{itemize}
    \item Scenario-$I$: It is assumed that eavesdropper ($\mathcal{E}1$) is located near $\mathcal{R}$. As a result, $\mathcal{E}1$ tries to receive confidential information via RF link. It is important to mention that $\mathcal{E}1$ utilizes different RIS ($\mathcal{I}_{\mathcal{E}1}$) to obtain the secret information from $\mathcal{S}$.
    
    \item Scenario-$II$: Unlike Scenario-$I$, eavesdropper ($\mathcal{E}2$) is located near $\mathcal{D}$, therefore utilizes UOWC link to thieve information. In this case, $\mathcal{E}2$ is capable of intercepting information via $\mathcal{I}_{\mathcal{E}2}$ from $\mathcal{D}$.  
    
    \item Scenario-$III$: In this scenario, both the eavesdroppers ($\mathcal{E}1$ and $\mathcal{E}2$) can intercept private information utilizing RF and UOWC links simultaneously. 
\end{itemize}

In the proposed model, $\mathcal{S}$ and $\mathcal{E}1$ each comprise a single antenna, whereas $\mathcal{R}$ functions as a transceiver with one transmitting and one receiving antenna. Conversely, $\mathcal{D}$ has a single photo detector to receive optical signals, whereas $\mathcal{I}_{\mathcal{R}}$, $\mathcal{I}_{\mathcal{D}}$, $\mathcal{I}_{\mathcal{E}1}$, and $\mathcal{I}_{\mathcal{E}2}$ are equipped with $N_{t}$, $N_{d}$, $N_{e_{s}}$, and $N_{e_{r}}$ reflecting elements, respectively. During the first hop, $\mathcal{S}$ transmits signal to $\mathcal{R}$ in the presence of $\mathcal{E}1$. We assume all RF links (i.e., $\mathcal{S}-\mathcal{I}_{\mathcal{R}}-\mathcal{R}$ and $\mathcal{S}-\mathcal{I}_{\mathcal{E}1}-\mathcal{E}1$) experience $\alpha$-$\mu$ fading distributions. It is noteworthy that $\mathcal{R}$ is used to convert the received RF signal into the corresponding optical format and subsequently transmitting it towards $\mathcal{D}$ under the influence of $\mathcal{E}2$. Here, we assume the UOWC links (i.e., $\mathcal{R}-\mathcal{I}_{\mathcal{D}}-\mathcal{D}$ and $\mathcal{R}-\mathcal{I}_{\mathcal{E}2}-\mathcal{E}2$) are subjected to mEGG distribution.

\subsection{SNRs of Individual Links}
For scenario-$I$, the channel gains for $\mathcal{S}-\mathcal{I}_{\mathcal{R}}$ and $\mathcal{S}-\mathcal{I}_{\mathcal{E}1}$ links are denoted by $g_{l_{1}} \left(l_{1}=1,2,3,...,N_{t}\right)$ and $f_{l_{2}}\left(l_{2}=1,2,3,...,N_{e_{s}}\right)
$, respectively. Therefore, the channel gains for $\mathcal{I}_{\mathcal{R}}-\mathcal{R}$ and $\mathcal{I}_{\mathcal{E}1}-\mathcal{E}1$ links are denoted by $p_{l_{1}}$ and $q_{l_{2}}$, respectively.
In terms of these relevant channels, $g_{l_{1}}=\nu_{l_{1}}e^{i\Phi_{l_{1}}}$, $f_{l_{2}}=\beta_{l_{2}}e^{i\Phi_{l_{2}}}$,
$p_{l_{1}}=\eta_{l_{1}}e^{i\varphi_{l_{1}}}$ and 
$q_{l_{2}}=\delta_{l_{2}}e^{i\varphi_{l_{2}}}$,
where, $\nu_{l_{1}}$, $\beta_{l_{2}}$,
$\eta_{l_{1}}$ and $\delta_{l_{2}}$ are $\alpha-\mu$ distributed random variables (RVs), $\Phi_{l_{1}}$, $\Phi_{l_{2}}$, $\varphi_{l_{1}}$ and $\varphi_{l_{2}}$ are the identical phases of channel gains for the received signal. Additionally, the RISs produce controllable phases $\theta_{l_{1}}$ and $\Theta_{l_{2}}$ using the reflecting elements indexed as $l_{1}$ and $l_{2}$, respectively. The transmitted signal, denoted as x, originates from $\mathcal{S}$ and carries a power $P_{s}$, $w_{r}\sim\Tilde{\mathcal{C}}(0,W_{r})$ and $w_{e_{1}}\sim\Tilde{\mathcal{C}}(0,W_{e_{1}})$ are the additive white Gaussian noises at $\mathcal{R}$ and $\mathcal{E}1$ with noise power $W_{r}$ and $W_{e_{1}}$, respectively. 
Now, the instantaneous SNR at $\mathcal{R}$ and $\mathcal{E}1$ can be written as
\begin{align}
    \label{snr5}
\gamma_{t}=&\left[\sum_{l_{1}=0}^{N_{t}}\nu_{l_{1}}\eta_{l_{1}}e^{i(\theta_{l_{1}}-\Phi_{l_{1}}-\varphi_{l_{1}})}\right]^{2}\Omega_{t},
\\
   \label{snr6}
\gamma_{e_{s}}=&\left[\sum_{l_{2}=0}^{N_{e_{s}}}\beta_{l_{2}}\delta_{l_{2}}e^{i(\Theta_{l_{2}}-\Phi_{l_{2}}-\varphi_{l_{2}})}\right]^{2}\Omega_{e_{s}},
\end{align}
where $\Omega_{t}=\frac{P_{s}}{W_{r}}$ and $\Omega_{e_{s}}=\frac{P_{s}}{W_{e_{1}}}$ are the average SNRs of $\mathcal{S}-\mathcal{I}_{\mathcal{R}}-\mathcal{R}$ and $\mathcal{S}-\mathcal{I}_{\mathcal{E}1}-\mathcal{E}1$ links, respectively. 
 In order to maximize the instantaneous SNR, the optimal values of $\theta_{l_{1}}$ and $\Theta_{l_{2}}$ are chosen as $\theta_{l_{1}}=\Phi_{l_{1}}+\varphi_{l_{1}}$ and $\Theta_{l_{2}}=\Phi_{l_{2}}+\varphi_{l_{2}}$. As a result, the maximum achievable SNR at $\mathcal{R}$ and $\mathcal{E}1$ can be denoted as
\begin{align}
    \label{snr7}
\gamma_{t,max}=&\left(\sum_{l_{1}=0}^{N_{t}}\nu_{l_{1}}\eta_{l_{1}}\right)\Omega_{t},
\\
   \label{sn8}
\gamma_{e_{s},max}=&\left(\sum_{l_{2}=0}^{N_{e_{s}}}\beta_{l_{2}}\delta_{l_{2}}\right)\Omega_{e_{s}}.
\end{align}
For scenario-$II$, the identical form as \eqref{snr7} and \eqref{sn8} can be utilized to express the maximum instantaneous SNR at $\mathcal{D}$ and $\mathcal{E}2$ as
\begin{align}
    \label{snr8}
\gamma_{d,max}&=\left(\sum_{l_{3}=0}^{N_{d}}\rho_{l_{3}}\psi_{l_{3}}\right)\Omega_{r_{d}},
\\
   \label{sn9}
\gamma_{e_{r},max}&=\left(\sum_{l_{4}=0}^{N_{e_{r}}}\varrho_{l_{4}}\Psi_{l_{4}}\right)\Omega_{r_{e_{r}}},
\end{align}
where $\rho_{l_{3}}$, $\psi_{l_{3}}$, $\varrho_{l_{4}}$ and $\Psi_{l_{4}}$ are the mEGG distributed random variables (RVs), $\Omega_{r_{d}}$ and $\Omega_{r_{e_{r}}}$ are the average SNRs of $\mathcal{R}-\mathcal{I}_{\mathcal{D}}-\mathcal{D}$ and $\mathcal{R}-\mathcal{I}_{\mathcal{E}2}-\mathcal{E}2$ links, respectively. Hence, assuming variable gain amplify-and-forward (AF) relaying, the instantaneous SNR of the proposed model  can be written as 
\begin{align}
    \gamma_{eq}\simeq min\{\gamma_{t},\gamma_{d}\}.
\end{align}

\subsection{Statistical Characteristics of RF Links}
We assume all the RF links (i.e., $\mathcal{S}-\mathcal{I}_{\mathcal{R}}-\mathcal{R}$ and $\mathcal{S}-\mathcal{I}_{\mathcal{E}1}-\mathcal{E}1$) are subjected to $\alpha-\mu$ fading distributions. Therefore, the PDF and CDF of $\gamma_{j}$ can be written, respectively as \cite{d1}
\begin{align}
  \label{t1}
f_{\gamma_{j}}(\gamma)=&\frac{(\frac{l_{j}}{\Omega_{j}})^{\frac{u_{j}}{2}} }{2\Gamma(u_{j})v_{j}^{u_{j}}}\gamma^{\frac{u_{j}}{2}-1}G_{0,1}^{1,0}\left[\sqrt{\frac{l_{j}\gamma}{\Omega_{j}v_{j}^{2}}}\biggl | 
    \begin{array}{c}
     -\\
     0 \\ 
    \end{array}
    \right],
\\
\label{t2}
   F_{\gamma_{j}}(\gamma)=&1-\sum_{\eta_{j}=0}^{u_{j}-1}\frac{\gamma^{\frac{\eta_{j}}{2}}}{\eta_{j}!}\left(\sqrt{\frac{l_{j}}{\Omega_{j}v_{j}^{2}}}\right)^{\eta_{j}}G_{0,1}^{1,0}\left[\sqrt{\frac{l_{j}\gamma}{\Omega_{j}v_{j}^{2}}}\biggl | 
    \begin{array}{c}
     -\\
     0 \\ 
    \end{array}
    \right],
\end{align}
where j $\in$ \{t, $e_{s}$\},
t and $e_{s}$ represent the $\mathcal{S}-\mathcal{I}_{\mathcal{R}}-\mathcal{R}$ and $\mathcal{S}-\mathcal{I}_{\mathcal{E}1}-\mathcal{E}1$ links, respectively, $l_{j}$ denotes the path loss, $\Omega_j$ denotes the average SNR, $G_{p,q}^{m,n}\left[\cdot \right]$ represents the Meijer's $G$ function as defined in \cite{art5}. Now, the shape parameter, $u_{j}$ and the scale parameter, $v_{j}$ can be calculated as
$u_{j}=\frac{N_{j}\mathbb{E}(M_{j})^{2}}{Var(M_{j})}$
and 
$v_{j}=\frac{Var(M_{j})}{{\mathbb{E}(M_{j})}}$,
where $N_j$ denotes the number of reflecting elements of the RISs. Here, $\mathbb{E}(M_{j})$ is associated with the $k^{th}$ moment of $M_{j}$. Assuming $M_{j}$ being the random variable, the $k^{th}$ moment of $M_{j}$ can be expressed as \eqref{E1}, shown at the top of this page,
where $\alpha_{j,s}$ and $\mu_{j,s}$ denote the fading parameters, and $\phi_{j,s}$ denotes the scale parameter due to the $\mathcal{S}-\mathcal{I}_\mathcal{R}$ and $\mathcal{S}-\mathcal{I}_{\mathcal{E}1}$ links, respectively. Similarly, $\alpha_{j,r}$, $\mu_{j,r}$, and $\phi_{j,r}$ have the same definitions as stated earlier but for the $\mathcal{I}_\mathcal{R} - \mathcal{R}$ and $\mathcal{I}_{\mathcal{E}1} - \mathcal{E}1$ link, respectively. Additionally, the four constants (i.e., $\lambda_1$, $\lambda_2$, $\lambda_3$, and $\lambda_4$) assumed in \eqref{E1}, are expressed as follows: 
\begin{align} \nonumber
    &\lambda_{1}=\mu_{j,s}-\frac{\alpha_{j,s}\mu_{j,s}-\alpha_{j,r}\mu_{j,r}}{2\alpha_{j,s}}-\frac{2k+\alpha_{j,s}\mu_{j,s}+\alpha_{j,r}\mu_{j,r}}{4},
    \\ \nonumber
    & \lambda_{2}=\mu_{j,r}+\frac{\alpha_{j,s}\mu_{j,s}-\alpha_{j,r}\mu_{j,r}}{2\alpha_{j,s}}-\frac{2k+\alpha_{j,s}\mu_{j,s}+\alpha_{j,r}\mu_{j,r}}{4},
\end{align}
\vspace{-5mm}
\begin{align} \nonumber
\lambda_{3}=&\frac{(\alpha_{j,s}\mu_{j,s}-\alpha_{j,r}\mu_{j,r})}{2}+\frac{\alpha_{j,s}(2k+\alpha_{j,s}\mu_{j,s}+\alpha_{j,r}\mu_{j,r})}{4}
\\ \nonumber
&-\alpha_{j,s}\mu_{j,s},
\end{align}
\vspace{-5mm}
\begin{align} \nonumber
\lambda_{4}=&\frac{\alpha_{j,r}(2k+\alpha_{j,r}\mu_{j,r}+\alpha_{j,s}\mu_{j,s})}{4}-\frac{(\alpha_{j,s}\mu_{j,s}-\alpha_{j,r}\mu_{j,r})}{2}
\\ \nonumber
&-\alpha_{j,r}\mu_{j,r}.
\end{align}
\setcounter{eqnback}{\value{equation}}
\setcounter{equation}{9}
\begin{figure*}
\begin{align}
\label{E1}
    \mathbb{E}(M_{j}^{k})=&\frac{\alpha_{j,r}(\mu_{j,s})^{\lambda_{1}}(\mu_{j,r})^{\lambda_{2}}
(\phi_{j,s})^{\lambda_{3}}(\phi_{j,r})^{\lambda_{4}}}{2\Gamma(\mu_{j,s})\Gamma(\mu_{j,r})}
\Gamma\left(\frac{\alpha_{j,s}\mu_{j,s}+\alpha_{j,r}\mu_{j,r}}{4}+\frac{\alpha_{j,s}\mu_{j,s}-\alpha_{j,r}\mu_{j,r}}{2\alpha_{j,s}}+\frac{k}{2}\right)
\nonumber
\\  
\times&\Gamma\left(\frac{\alpha_{j,s}\mu_{j,s}+\alpha_{j,r}\mu_{j,r}}{4}-\frac{\alpha_{j,s}\mu_{j,s}-\alpha_{j,r}\mu_{j,r}}{2\alpha_{j,s}}+\frac{k}{2}\right),
\end{align}
\hrulefill
\end{figure*}
\setcounter{eqnback}{\value{equation}}
\setcounter{equation}{12}
\begin{figure*}[!b]
\hrulefill
\begin{align}
    \label{E2}
    \mathbb{E}(M_{n}^{k})&=\frac{\omega_{n,s}\omega_{n,r}\xi_{n,s}^{2}\xi_{n,r}^{2}(A_{n,s}A_{n,r}\lambda_{n,s}\lambda_{n,r})^{k}\prod_{\iota=1}^{4}\Gamma(k+R_{1n,\iota})}{\prod_{\iota=1}^{2}\Gamma(k+Q_{1n,\iota})}
 \nonumber
 \\    
 &+\frac{\omega_{n,s}(1-\omega_{n,r})\xi_{n,s}^{2}\xi_{n,r}^{2}(A_{n,s}A_{n,r}b_{n,r}c_{n,r}\lambda_{n,s})^{kc_{n,r}}\prod_{\iota=1}^{4}\Gamma\left(k+R_{2n,\iota}\right)}{(2\pi)^{\frac{1}{2}(c_{n,r}-1)}c_{n,r}^{\frac{1}{2}}\Gamma(a_{n,r})\prod_{\iota=1}^{2}\Gamma\left(k+Q_{2n,\iota}\right)}
    \nonumber
    \\
 &+\frac{\omega_{n,r}(1-\omega_{n,s})\xi_{n,s}^{2}\xi_{n,r}^{2}(A_{n,r}A_{n,s}b_{n,s}c_{n,s}\lambda_{n,r})^{kc_{n,s}}\prod_{\iota=1}^{4}\Gamma(k+R_{3n,\iota})}{(2\pi)^{\frac{1}{2}(c_{n,s}-1)}c_{n,s}^{\frac{1}{2}}\Gamma(a_{n,s})\prod_{\iota=1}^{2}\Gamma(k+Q_{3n,\iota})}
\nonumber     
\\     
&+\frac{(1-\omega_{n,s})(1-\omega_{n,r})\xi_{n,s}^{2}\xi_{n,r}^{2}(A_{n,s}b_{n,s})^{kc_{n,s}}(A_{n,r}b_{n,r})^{kc_{n,r}}\prod_{\iota=1}^{4}\Gamma\left(k+R_{4n,\iota}\right)}{\Gamma(a_{n,s})\Gamma(a_{n,r})\prod_{\iota=1}^{2}\Gamma\left(k+Q_{4n,\iota}\right)},
\end{align}
\end{figure*}
\begin{remark}
    In \cite{ibrahim2023effective}, a comprehensive analysis of $\alpha$-$\mu$ fading channels is presented, proposing the utilization of $\alpha$-$\mu$ distributions in RF links to provide the advantages observed in modeling other multipath fading channels. Moreover, by assuming $\alpha$-$\mu$ channels, we can incorporate many other existing works \cite{art16,ahmed2023enhancing}.
\end{remark}
\subsection{Statistical Characteristics of UOWC Links}   
We consider all the UOWC links to experience mEGG distributions. Therefore, the PDF and CDF of UOWC links can be written as \cite{10413214}
\setcounter{eqnback}{\value{equation}}
\setcounter{equation}{10}
\begin{align}
\label{d1}
f_{\gamma_{n}}(\gamma)=&\frac{\mathcal{M}_{1,n}}{r_{n}}\gamma^{\frac{u_{n}}{r_{n}}-1}G_{0,1}^{1,0}\left[\frac{\mathbb{E}(M_{n})}{v_{n}\Omega_{r_{n}}^{\frac{1}{r_{n}}}}\gamma^{\frac{1}{r_{n}}}\biggl | 
    \begin{array}{c}
     - \\
     0 \\ 
    \end{array}
    \right],
\\
\label{d2}
F_{\gamma_{n}}
(\gamma)=&\frac{\mathcal{M}_{1,n}\gamma^{\frac{u_{n}}{r_{n}}}}{\sqrt{r_{n}}(2\pi)^{\frac{r_{n}-1}{2}}}
\nonumber
\\
\times &G_{1,r_{n}+1}^{r_{n},1}\left[\frac{\mathbb{E}(M_{n})^{r_{n}}\gamma}{\Omega_{r_{n}}(v_{n}r_{n})^{r_{n}}}\biggl | 
    \begin{array}{c}
     1-\frac{u_{n}}{r_{n}} \\
     0,\frac{r_{n}-1}{r_{n}},\frac{-u_{n}}{r_{n}} \\ 
    \end{array}
    \right],
\end{align}
where $n$ $\in$ \{$d$, $e_{r}$\},
$d$ and $e_{r}$ denote the links corresponding to $\mathcal{R}-\mathcal{I}_{\mathcal{D}}-\mathcal{D}$ and $\mathcal{R}-\mathcal{I}_{\mathcal{E}2}-\mathcal{E}2$, respectively, $\mathcal{M}_{1,n}=\frac{\mathbb{E}(M_{n})^{u_{n}}}{\Gamma(u_{n})v_{n}^{u_{n}}\mu_{r_{n}}^{\frac{u_{n}}{r_{n}}}}$, $\Omega_{r_{n}}$ is the electrical SNR which is defined as $\Omega_{r_{n}}=\frac{(\eta_{0}\mathbb{E}(M_{n}))^{r_{n}}}{\epsilon}$, $r_{n}$ represents the receiver detection technique, $r_{n}=1$ signifies the heterodyne detection (HD), whereas $r_{n}=2$ denotes the intensity modulation/direct detection (IM/DD) technique, $\eta_{o}$ is the conversion co-efficient parameter, $\epsilon$ is related to variance of AWGN. Now, the shape parameter ($u_{n}$) and the scale parameter ($v_n$) of the UOWC links can be expressed, respectively as $u_{n}=\frac{N_{n}\mathbb{E}(M_{n})^{2}}{Var(M_{n})}$
    and
$v_{n}=\frac{Var(M_{n})}{{\mathbb{E}(M_{n})}}$, where $N_{n}$ denotes the number of reflecting elements in the RIS. Now, considering $M_{n}$ as the random variable, the $k^{th}$ moment of $M_{n}$ can be written as \eqref{E2}, where $a_{n,s}, b_{n,s}$, and $c_{n,s}$ denote the GG distribution parameters, $\lambda_{n,s}$ represent the exponential distribution parameter,  $\omega_{n,s}$ with $0<\omega_{n,s}<1$ denote the mixture weight parameter, and $\xi_{n,s}$ corresponds to the pointing error due to the $\mathcal{R}-\mathcal{I}_{\mathcal{D}}$ and $\mathcal{R}-\mathcal{I}_{\mathcal{E}2}$, respectively. Again, $a_{n,r}, b_{n,r}$, $c_{n,r}$, $\lambda_{n,r}$, $\omega_{n,r}$, and $\xi_{n,r}$ have the similar definitions as mentioned earlier but for the $\mathcal{I}_{\mathcal{D}}-\mathcal{D}$ and $\mathcal{I}_{\mathcal{E}2}-\mathcal{E}2$, respectively. Now, the constants associated with the $\mathbb{E}(M_{n})$ can be given as follows:
$R_{1n}=\begin{bmatrix}
    1 & \xi_{n,r}^{2} & 1 & \xi_{n,s}^{2}
\end{bmatrix}$, $R_{2n}=\begin{bmatrix}
    a_{n,r} & \frac{\xi_{n,r}^{2}}{c_{n,r}} & 1-\mathcal{M}_{2,n} & 1-\mathcal{M}_{3,n}
\end{bmatrix}$, $R_{3n}=\begin{bmatrix}
    a_{n,s} & \frac{\xi_{n,s}^{2}}{c_{n,s}} & \mathcal{M}_{4,n} & \mathcal{M}_{5,n}
\end{bmatrix}$, $R_{4n}=\begin{bmatrix}
    a_{n,r} & \frac{\xi_{n,r}^{2}}{c_{r}} & a_{n,s} & \frac{\xi_{n,s}^{2}}{c_{s}}
\end{bmatrix}$, $Q_{1n}=\begin{bmatrix}
    \xi_{n,r}^{2}+1 & \xi_{n,s}^{2}+1
\end{bmatrix}$, $Q_{2n}=\begin{bmatrix}
    \frac{\xi_{n,r}^{2}}{c_{n,r}}+1 & 1-\mathcal{M}_{6,n}
\end{bmatrix}$, $Q_{3n}=\begin{bmatrix}
    \frac{\xi_{n,s}^{2}}{c_{n,s}}+1 & \mathcal{M}_{7,n}
\end{bmatrix}$, and $Q_{4n}=\begin{bmatrix}
    \frac{\xi_{n,r}^{2}}{c_{n,r}}+1 & \frac{\xi_{n,s}^{2}}{c_{n,s}}+1
\end{bmatrix}$, $\mathcal{M}_{2,n}=0,...,\frac{c_{n,r}-1}{c_{n,r}}$, 
$\mathcal{M}_{3,n}=\frac{1-\xi_{n,s}^{2}}{c_{n,r}},...,\frac{c_{n,r}-\xi_{n,s}^{2}}{c_{n,r}}$, 
$\mathcal{M}_{4,n}=\frac{1}{c_{n,s}},...,1$,
$\mathcal{M}_{5,n}=\frac{\xi_{n,r}^{2}}{c_{n,s}},...,\frac{\xi_{n,r}^{2}+c_{n,s}-1}{c_{n,s}}$,
$\mathcal{M}_{6,n}=-\frac{\xi_{n,s}^{2}}{c_{n,r}},...,\frac{c_{n,r}-\xi_{n,s}^{2}-1}{c_{n,r}}$,
$\mathcal{M}_{7,n}=\frac{1+\xi_{n,r}^{2}}{c_{n,s}},...,\frac{c_{n,s}+\xi_{n,r}^{2}}{c_{n,s}}$.
\setcounter{eqnback}{\value{equation}}
\setcounter{equation}{14}
\begin{figure*}[!t]
\begin{align}
    \label{feq}
    F_{\gamma_{eq}}(\gamma)=&1-\sum_{\eta_{t}=0}^{u_{t}-1}\frac{\gamma^{\frac{\eta_{t}}{2}}}{\eta_{t}!}\left(\sqrt{\frac{l_{t}}{\Omega_{t}v_{t}^{2}}}\right)^{\eta_{t}}G_{0,1}^{1,0}\left[\sqrt{\frac{l_{t}\gamma}{\Omega_{t}v_{t}^{2}}}\biggl | 
    \begin{array}{c}
     -\\
     0 \\ 
    \end{array}
    \right]+\sum_{\eta_{t}=0}^{u_{t}-1}\frac{\mathcal{M}_{1}}{\eta_{t}!\sqrt{r_{d}}(2\pi)^{\frac{r_{d}-1}{2}}}\left(\sqrt{\frac{l_{t}}{\Omega_{t}v_{t}^{2}}}\right)^{\eta_{t}}\gamma^{\frac{\eta_{t}}{2}+\frac{u_{d}}{r_{d}}}
\nonumber
\\    \times&G_{0,1}^{1,0}\left[\sqrt{\frac{l_{t}\gamma}{\Omega_{t}v_{t}^{2}}}\biggl | 
    \begin{array}{c}
     -\\
     0 \\ 
    \end{array}
    \right]G_{1,r_{d}+1}^{r_{d},1}\left[\frac{\mathbb{E}(M_{d})^{r_{d}}\gamma}{\Omega_{r_{d}}(v_{d}r_{d})^{r_{d}}}\biggl | 
    \begin{array}{c}
     1-\frac{u_{d}}{r_{d}} \\
     0,\frac{r_{d}-1}{r_{d}},\frac{-u_{d}}{r_{d}} \\ 
    \end{array}
    \right].
\end{align}
\hrulefill
\end{figure*}
\setcounter{eqnback}{\value{equation}}
\setcounter{equation}{18}
\begin{figure*}[!t]
\begin{align}
    \label{feqasy}
    F_{\gamma_{eq}}^{\infty}(\gamma)\approx&\Biggl(\sqrt{\frac{l_{t}\gamma}{\Omega_{t}v_{t}^{2}}}\Biggl)^{u_{t}}\frac{1}{u_{t}!}+\sum_{i=1}^{r_{d}}\frac{\mathcal{M}_{1}\gamma^{\frac{u_{d}}{r_{d}}}}{\sqrt{r_{d}}(2\pi)^{\frac{r_{d}-1}{2}}}\Biggl(\frac{\mathbb{E}(M_{d})^{r_{d}}\gamma}{\Omega_{r_{d}}(v_{d}r_{d})^{r_{d}}}\Biggl)^{P_{i}-1}\frac{\prod_{\kappa=1;\kappa\neq i}^{r_{d}}(P_{i}-P_{\kappa})\Gamma(1+\frac{u_{d}}{r_{d}}-P_{i})}{\Gamma(1+P_{r_{d}+1}-P_{i})}.
\end{align}
\hrulefill
\end{figure*}
\begin{remark}
    The newly proposed EGG model exhibits strong agreement with measured data across various atmospheric turbulence conditions, encompassing both Exponential-Gamma (EG) and Exponential-Lognormal distributions \cite{zedini2019unified}. This robust alignment establishes the EGG distribution as the most fitting probability distribution for characterizing fluctuations in underwater optical signal irradiance caused by air bubbles and temperature-induced turbulence. For instance, EGG model can be transformed into EG distribution by simply setting $c=1$ as shown in \cite{art41}.
\end{remark}

\subsection{CDF of SNR for Dual-hop RF-UOWC Link}
The CDF of $\gamma_{eq}$ is defined as
\setcounter{eqnback}{\value{equation}}
\setcounter{equation}{13}
\begin{align}
\label{c3}
 F_{\gamma_{eq}}=F_{\gamma_{t}}(\gamma)+F_{\gamma_{d}}(\gamma)-F_{\gamma_{t}}(\gamma)F_{\gamma_{d}}(\gamma).
\end{align}   
Now, substituting \eqref{t2} and \eqref{d2} into \eqref{c3} and after performing some algebraic calculations, $F_{\gamma_{eq}}$  can be expressed as shown in \eqref{feq}.

\noindent \textit{Asymptotic Analysis:}
For a more comprehensive understanding, it is essential to examine the asymptotic behavior of the proposed system. The asymptotic CDF expression due to the RF link can be obtained as \cite{n2}
\setcounter{eqnback}{\value{equation}}
\setcounter{equation}{15}
\begin{align}
\label{f1}
    F_{\gamma_{t}}^{\infty}(\gamma)=\Biggl(\sqrt{\frac{l_{t}\gamma}{\Omega_{t}v_{t}^{2}}}\Biggl)^{u_{t}}\frac{1}{u_{t}!}.
\end{align}
On the other hand, by applying the identity of \cite[Eq.~(8.2.2.14)]{art7} and utilizing the formula of \cite[Eq.~(41)]{n1} to expand the Meijer'G function, the asymptotic CDF of UOWC link can be derived finally as
\begin{align}
\label{f2}
    F_{\gamma_{d}}^{\infty}(\gamma)=&\sum_{i=1}^{r_{d}}\frac{\mathcal{M}_{1,n}\gamma^{\frac{u_{d}}{r_{d}}}}{\sqrt{r_{d}}(2\pi)^{\frac{r_{d}-1}{2}}}\Biggl(\frac{\mathbb{E}(M_{d})^{r_{d}}\gamma}{\Omega_{r_{d}}(v_{d}r_{d})^{r_{d}}}\Biggl)^{P_{i}-1}
\nonumber
\\
    &\frac{\prod_{\kappa=1;\kappa\neq i}^{r_{d}}(P_{i}-P_{\kappa})\Gamma(1+\frac{u_{d}}{r_{d}}-P_{i})}{\Gamma(1+P_{r_{d}+1}-P_{i})},
\end{align}
where $P=\begin{bmatrix}
    1 & \frac{1}{r_{d}} & 1+\frac{u_{d}}{r_{d}} 
\end{bmatrix}$. 
Hence, the asymptotic CDF for the proposed mixed model can be expressed finally as
\begin{align}
\label{fe}
    F_{\gamma_{eq}}^{\infty}(\gamma)=F_{\gamma_{t}}^{\infty}(\gamma)+F_{\gamma_{d}}^{\infty}(\gamma).
\end{align}
Now, substituting \eqref{f1} and \eqref{f2} into \eqref{fe}, the asymptotic expression is obtained as shown in \eqref{feqasy}.

\section{SECRECY PERFORMANCE ANALYSIS}
\label{spa}
In this section, we illustrate the closed-form analytical representations of diverse secrecy metrics, encompassing the lower bound of SOP, ASC, probability of SPSC, and EST analysis. Additionally, we perform asymptotic analyses to delve deeper into the understanding of the high SNR region facilitated by the proposed RIS-aided model.
\subsection{ASC Analysis}
ASC serves as a crucial metric for evaluating the effectiveness of secure communication, representing the mean value of instantaneous secrecy capacity. Assuming the eavesdropper can intercept information through the RF link, ASC can be mathematically defined as 
\cite[Eq. (15)]{art6}
\setcounter{eqnback}{\value{equation}}
\setcounter{equation}{19}
\begin{align}
\label{asc}
    ASC^{I}=\int_{0}^{\infty}\frac{ F_{\gamma_{e_{s}}}(\gamma)}{1+\gamma}[1-F_{\gamma_{eq}}(\gamma)]d\gamma.
\end{align}

\begin{theorem}
The ASC can be expressed in closed form as 
\begin{align} \label{ascf}
ASC^{I}=&\sum_{\eta_{t}=0}^{u_{t}-1}\frac{1}{\eta_{t}!}\left(\sqrt{\frac{l_{t}}{\Omega_{t}v_{t}^{2}}}\right)^{\eta_{t}}
\Biggr[\mathcal{X}_{1}-\frac{\mathcal{M}_{1}}{\sqrt{r_{d}}(2\pi)^{\frac{r_{d}-1}{2}}}\mathcal{X}_{2}
\nonumber
\\
-&\sum_{\eta_{e_{s}}=0}^{u_{e_{s}}-1}\frac{1}{\eta_{e_{s}}!}\left(\sqrt{\frac{l_{e_{s}}}{\Omega_{e_{s}}v_{e_{s}}^{2}}}\right)^{\eta_{e_{s}}}\mathcal{X}_{3}
\nonumber
\\
+&\sum_{\eta_{e_{s}}=0}^{u_{e_{s}}-1}\frac{\mathcal{M}_{1}r_{d}^{-\frac{1}{2}}}{\eta_{e_{s}}!(2\pi)^{\frac{r_{d}-1}{2}}}\left(\sqrt{\frac{l_{e_{s}}}{\Omega_{e_{s}}v_{e_{s}}^{2}}}\right)^{\eta_{e_{s}}}\mathcal{X}_{4}\Biggr].
\end{align}
where $\mathcal{X}_{1}$, $\mathcal{X}_{2}$, $\mathcal{X}_{3}$, and $\mathcal{X}_{4}$ are the integral terms.
\end{theorem}
\begin{proof}
The four integral terms (i.e., $\mathcal{X}_{1}$, $\mathcal{X}_{2}$, $\mathcal{X}_{3}$, and $\mathcal{X}_{4}$) can be obtained as follows:
\subsubsection{Derivation of $\mathcal{X}_{1}$}
$\mathcal{X}_{1}$ is denoted as
\begin{align}
\mathcal{X}_{1}=\int_{0}^{\infty}\gamma^{\frac{\eta_{t}}{2}}\frac{1}{1+\gamma}G_{0,1}^{1,0}\left[\sqrt{\frac{l_{t}}{\Omega_{t}v_{t}^{2}}}\gamma^{\frac{1}{2}}\biggl | 
    \begin{array}{c}
     -\\
     0 \\ 
    \end{array}
    \right]d\gamma.
    \nonumber
\end{align}
Using the identity \cite[Eq. (8.4.2.5)]{art7} to transform $\frac{1}{1+\gamma}$ into Meijer's G function and calculating the integral utilizing \cite[Eq. (2.24.1.1)]{art7}, $\mathcal{X}_{1}$ is determined as
\begin{align}
   \label{x1}\mathcal{X}_{1}=&\int_{0}^{\infty}\gamma^{\frac{\eta_{t}}{2}}G_{1,1}^{1,1}\left[\gamma\biggl | 
    \begin{array}{c}
     0\\
     0 \\ 
    \end{array}
    \right]G_{0,1}^{1,0}\left[\sqrt{\frac{l_{t}}{\Omega_{t}v_{t}^{2}}}\gamma^{\frac{1}{2}}\biggl | 
    \begin{array}{c}
     -\\
     0 \\ 
    \end{array}
    \right]d\gamma
\nonumber
\\
=&\frac{1}{\sqrt{\pi}}G_{1,3}^{3,1}\left[\frac{l_{t}}{4\Omega_{t}v_{t}^{2}}\biggl | 
    \begin{array}{c}
     \frac{-\eta_{t}}{2}\\
     0,\frac{1}{2},\frac{-\eta_{t}}{2} \\ 
    \end{array}
\right].
\end{align} 
\subsubsection{Derivation of $\mathcal{X}_{2}$}
$\mathcal{X}_{2}$ is denoted as
\begin{align}
\mathcal{X}_{2}=&\int_{0}^{\infty}\gamma^{\frac{\eta_{t}}{2}+\frac{u_{d}}{r_{d}}}\frac{1}{1+\gamma}G_{0,1}^{1,0}\left[\sqrt{\frac{l_{t}\gamma}{\Omega_{t}v_{t}^{2}}}\biggl | 
    \begin{array}{c}
     -\\
     0 \\ 
    \end{array}
    \right]
    \nonumber
    \\
\times&G_{1,r_{d}+1}^{r_{d},1}\left[\frac{\mathbb{E}(M_{d})^{r_{d}}}{\Omega_{r_{d}}(v_{d}r_{d})^{r_{d}}}\gamma\biggl | 
    \begin{array}{c}
     1-\frac{u_{d}}{r_{d}} \\
     0,\frac{r_{d}-1}{r_{d}},\frac{-u_{d}}{r_{d}} \\ 
    \end{array}
    \right]d\gamma.
    \nonumber
\end{align}
Utilizing the same identity as in $\mathcal{X}_{1}$ to convert $\frac{1}{1+\gamma}$ term into Meijer's G function, all the Meijer's G terms are transformed into  Fox's H function according to\cite[Eq. (6.2.8)]{art8}. Now, applying the identity of \cite[Eq. (2.3)]{art9}, $\mathcal{X}_{2}$ is derived as
\begin{align}
\label{x2}
\mathcal{X}_{2}=&\int_{0}^{\infty}\gamma^{\frac{\eta_{t}}{2}+\frac{u_{d}}{r_{d}}}G_{1,1}^{1,1}\left[\gamma\biggl | 
    \begin{array}{c}
     0\\
     0 \\ 
    \end{array}
    \right]G_{0,1}^{1,0}\left[\sqrt{\frac{l_{t}\gamma}{\Omega_{t}v_{t}^{2}}}\biggl | 
    \begin{array}{c}
     -\\
     0 \\ 
    \end{array}
    \right]
    \nonumber
    \\
\times &G_{1,r_{d}+1}^{r_{d},1}\left[\frac{\mathbb{E}(M_{d})^{r_{d}}}{\Omega_{r_{d}}(v_{d}r_{d})^{r_{d}}}\gamma\biggl | 
    \begin{array}{c}
     1-\frac{u_{d}}{r_{d}} \\
     0,\frac{r_{d}-1}{r_{d}},\frac{-u_{d}}{r_{d}} \\ 
    \end{array}
    \right]d\gamma
\nonumber
\\
=&\int_{0}^{\infty}\gamma^{\frac{\eta_{t}}{2}+\frac{u_{d}}{r_{d}}}H_{1,1}^{1,1}\left[\gamma\biggl | 
    \begin{array}{c}
     (0,1)\\
     (0,1) \\ 
    \end{array}
    \right]
    \nonumber
    \\    \times &H_{0,1}^{1,0}\left[\sqrt{\frac{l_{t}\gamma}{\Omega_{t}v_{t}^{2}}}\biggl | 
    \begin{array}{c}
     -\\
     (0,1) \\ 
    \end{array}
    \right]
    \nonumber
    \\
\times&H_{1,r_{d}+1}^{r_{d},1}\left[\frac{\mathbb{E}(M_{d})^{r_{d}}\gamma}{\Omega_{r_{d}}(v_{d}r_{d})^{r_{d}}}\biggl | 
    \begin{array}{c}
     (1-\frac{u_{d}}{r_{d}},1) \\
     (0,1),(\frac{r_{d}-1}{r_{d}},1),(\frac{-u_{d}}{r_{d}},1) \\ 
    \end{array}
    \right]d\gamma
\nonumber
\\
=&2\Xi_{1}H_{1,0;1,1;1,r_{d}+1}^{1,0;1,1;r_{d},1}\left[
    \begin{array}{c}
     \Xi_{2}\\
     1;- \\ 
    \end{array}\biggl |    \begin{array}{c}
     (0,1)\\
     (0,1)
    \end{array}\biggl | \begin{array}{c}
     \Xi_{3}\\
     \Xi_{4}
    \end{array}\biggl | \begin{array}{c}
     \Xi_{5}, \Xi_{6} \\ 
    \end{array}  
    \right],
\end{align}
where $\Xi_{1}=\left(\sqrt{\frac{l_{t}}{\Omega_{t}v_{t}^{2}}}\right)^{-2-\eta_{t}-\frac{2u_{d}}{r_{d}}}$,
$\Xi_{2}=(-1-\eta_{t}-\frac{2u_{d}}{r_{d}};2,2)$,
$\Xi_{3}=(1-\frac{u_{d}}{r_{d}},1)$,
$\Xi_{4}=\left[(0,1),(\frac{r_{d}-1}{r_{d}},1),(\frac{-u_{d}}{r_{d}},1)\right]$,
$\Xi_{5}=\frac{\Omega_{t}v_{t}^{2}}{l_{t}}$, and
$\Xi_{6}=\frac{\Omega_{t}v_{t}^{2}\mathbb{E}(M_{d})^{r_{d}}}{l_{t}\Omega_{r_{d}}(v_{d}r_{d})^{r_{d}}}$.

\subsubsection{Derivation of $\mathcal{X}_{3}$}
$\mathcal{X}_{3}$ is denoted as
\begin{align}
\mathcal{X}_{3}=&\int_{0}^{\infty}\gamma^{\frac{\eta_{t}}{2}+\frac{\eta_{e_{s}}}{2}}\frac{1}{1+\gamma}G_{0,1}^{1,0}\left[\sqrt{\frac{l_{t}\gamma}{\Omega_{t}v_{t}^{2}}}\biggl | 
    \begin{array}{c}
     -\\
     0 \\ 
    \end{array}
    \right]
    \nonumber
    \\
\times&G_{0,1}^{1,0}\left[\sqrt{\frac{l_{e_{s}}\gamma}{\Omega_{e_{s}}v_{e_{s}}^{2}}}\biggl | 
    \begin{array}{c}
     -\\
     0 \\ 
    \end{array}
    \right]d\gamma.
    \nonumber
\end{align}
Utilizing the similar identities employed in $\mathcal{X}_{2}$, $\mathcal{X}_{3}$ is expressed as
\begin{align}
\label{x3}
\mathcal{X}_{3}=&\int_{0}^{\infty}\gamma^{\frac{\eta_{t}}{2}+\frac{\eta_{e_{s}}}{2}}G_{1,1}^{1,1}\left[\gamma\biggl | 
    \begin{array}{c}
     0\\
     0 \\ 
    \end{array}
    \right]G_{0,1}^{1,0}\left[\sqrt{\frac{l_{t}\gamma}{\Omega_{t}v_{t}^{2}}}\biggl | 
    \begin{array}{c}
     -\\
     0 \\ 
    \end{array}
    \right]
    \nonumber
    \\
\times&G_{0,1}^{1,0}\left[\sqrt{\frac{l_{e_{s}}\gamma}{\Omega_{e_{s}}v_{e_{s}}^{2}}}\biggl | 
    \begin{array}{c}
     -\\
     0 \\ 
    \end{array}
    \right]d\gamma
    \nonumber
    \\
=&\int_{0}^{\infty}\gamma^{\frac{\eta_{t}}{2}+\frac{\eta_{e_{s}}}{2}}H_{1,1}^{1,1}\left[\gamma\biggl | 
    \begin{array}{c}
     (0,1)\\
     (0,1) \\ 
    \end{array}
    \right]H_{0,1}^{1,0}\left[\sqrt{\frac{l_{t}\gamma}{\Omega_{t}v_{t}^{2}}}\biggl | 
    \begin{array}{c}
     -\\
     (0,1) \\ 
    \end{array}
    \right]
    \nonumber
    \\
\times&H_{0,1}^{1,0}\left[\sqrt{\frac{l_{e_{s}}\gamma}{\Omega_{e_{s}}v_{e_{s}}^{2}}}\biggl | 
    \begin{array}{c}
     -\\
     (0,1) \\ 
    \end{array}
    \right]d\gamma
    \nonumber
    \\
\nonumber
\\
=&2\Xi_{7}H_{1,0;1,1;0,1}^{1,0;1,1;1,0}\left[
    \begin{array}{c}
     \Xi_{8}\\
     1;- \\ 
    \end{array}\biggl |    \begin{array}{c}
     (0,1)\\
     (0,1)
    \end{array}\biggl | \begin{array}{c}
     -\\
     (0,1)
    \end{array}\biggl | \begin{array}{c}
     \Xi_{5}, \Xi_{9} \\ 
    \end{array}  
    \right],
\end{align}
where $\Xi_{7}=\left(\sqrt{\frac{l_{t}}{\Omega_{t}v_{t}^{2}}}\right)^{-2-\eta_{t}-\eta_{e_{s}}}$,
$\Xi_{8}=(-1-\eta_{t}-\eta_{e_{s}};2,1)$, and
$\Xi_{9}=\sqrt{\frac{l_{e_{s}}\Omega_{t}v_{t}^{2}}{l_{t}\Omega_{e_{s}}v_{e_{s}}^{2}}}$.
\subsubsection{Derivation of $\mathcal{X}_{4}$}
$\mathcal{X}_{4}$ is denoted as
\begin{align} \nonumber
\mathcal{X}_{4}&=\int_{0}^{\infty}\gamma^{\frac{\eta_{t}}{2}+\frac{\eta_{e_{s}}}{2}+\frac{u_{d}}{r_{d}}} 
\\ \nonumber
& \times G_{1,r_{d}+1}^{r_{d},1}\left[\frac{\mathbb{E}(M_{d})^{r_{d}}}{\Omega_{r_{d}}(v_{d}r_{d})^{r_{d}}}\gamma\biggl | 
    \begin{array}{c}
     1-\frac{u_{d}}{r_{d}} \\
     0,\frac{r_{d}-1}{r_{d}},\frac{-u_{d}}{r_{d}} \\ 
    \end{array}
    \right]
    \nonumber
    \\
& \times G_{0,1}^{1,0}\left[\left(\sqrt{\frac{l_{t}}{\Omega_{t}v_{t}^{2}}}+\sqrt{\frac{l_{e_{s}}}{\Omega_{e_{s}}v_{e_{s}}^{2}}}\right)\gamma^{\frac{1}{2}}\biggl | 
    \begin{array}{c}
     -\\
     0 \\ 
    \end{array}
    \right]\frac{1}{1+\gamma}d\gamma.
    \nonumber
\end{align}
Utilizing the similar identities as used in $\mathcal{X}_{2}$ and $\mathcal{X}_{3}$, $\mathcal{X}_{4}$ is finally derived as
\begin{align} \nonumber
\label{x4}
\mathcal{X}_{4}=&\int_{0}^{\infty}\gamma^{\frac{\eta_{t}}{2}+\frac{\eta_{e_{s}}}{2}+\frac{u_{d}}{r_{d}}}
 \nonumber
 \\
 \times &G_{1,r_{d}+1}^{r_{d},1}\left[\frac{\mathbb{E}(M_{d})^{r_{d}}}{\Omega_{r_{d}}(v_{d}r_{d})^{r_{d}}}\gamma\biggl | 
    \begin{array}{c}
     1-\frac{u_{d}}{r_{d}} \\
     0,\frac{r_{d}-1}{r_{d}},\frac{-u_{d}}{r_{d}} \\ 
    \end{array}
    \right]
 \nonumber
 \\    
\times &G_{0,1}^{1,0}\left[\left(\sqrt{\frac{l_{t}}{\Omega_{t}v_{t}^{2}}}+\sqrt{\frac{l_{e_{s}}}{\Omega_{e_{s}}v_{e_{s}}^{2}}}\right)\gamma^{\frac{1}{2}}\biggl | 
    \begin{array}{c}
     -\\
     0 \\ 
    \end{array}
    \right]
G_{1,1}^{1,1}\left[\gamma\biggl | 
    \begin{array}{c}
     0\\
     0 \\ 
    \end{array}
    \right]d\gamma
\nonumber
\\
=&\int_{0}^{\infty}\gamma^{\frac{\eta_{t}}{2}+\frac{\eta_{e_{s}}}{2}+\frac{u_{d}}{r_{d}}}H_{1,1}^{1,1}\left[\gamma\biggl | 
    \begin{array}{c}
     (0,1)\\
     (0,1) \\ 
    \end{array}
    \right]
\nonumber
\\
\times&H_{0,1}^{1,0}\left[\left(\sqrt{\frac{l_{t}}{\Omega_{t}v_{t}^{2}}}+\sqrt{\frac{l_{e_{s}}}{\Omega_{e_{s}}v_{e_{s}}^{2}}}\right)\gamma^{\frac{1}{2}}\biggl | 
    \begin{array}{c}
     -\\
     (0,1) \\ 
    \end{array}
    \right]
        \nonumber
    \\
\times&H_{1,r_{d}+1}^{r_{d},1}\left[\frac{\mathbb{E}(M_{d})^{r_{d}}}{\Omega_{r_{d}}(v_{d}r_{d})^{r_{d}}}\gamma\biggl | 
    \begin{array}{c}
     (1-\frac{u_{d}}{r_{d}},1) \\
     (0,1),(\frac{r_{d}-1}{r_{d}},1),(\frac{-u_{d}}{r_{d}},1) \\ 
    \end{array}
    \right]d\gamma
\nonumber
\\
=&2\Xi_{10}H_{1,0;1,1;1,r_{d}+1}^{1,0;1,1;r_{d},1}\left[
    \begin{array}{c}
     \Xi_{11}\\
     1;- \\ 
    \end{array}\biggl |    \begin{array}{c}
     (0,1)\\
     (0,1)
    \end{array}\biggl | \begin{array}{c}
     \Xi_{3}\\
     \Xi_{4}
    \end{array}\biggl | \begin{array}{c}
     \Xi_{12}, \Xi_{13} \\ 
    \end{array}  
    \right],
\end{align}
where $\Xi_{10}=\frac{1}{\left(\sqrt{\frac{l_{t}}{\Omega_{t}v_{t}^{2}}}+\sqrt{\frac{l_{e_{s}}}{\Omega_{e_{s}}v_{e_{s}}^{2}}}\right)^{2+\eta_{t}+\eta_{e_{s}}+\frac{2u_{d}}{r_{d}}}}$,
$\Xi_{11}=\left(-1-\eta_{t}-\eta_{e_{s}}-\frac{2u_{d}}{r_{d}};2,2\right)$,
$\Xi_{12}=\frac{1}{\left(\sqrt{\frac{l_{t}}{\Omega_{t}v_{t}^{2}}}+\sqrt{\frac{l_{e_{s}}}{\Omega_{e_{s}}v_{e_{s}}^{2}}}\right)^{2}}$, and 
$\Xi_{13}=\frac{\mathbb{E}(M_{d})^{r_{d}}}{\Omega_{r_{d}}(v_{d}r_{d})^{r_{d}}\left(\sqrt{\frac{l_{t}}{\Omega_{t}v_{t}^{2}}}+\sqrt{\frac{l_{e_{s}}}{\Omega_{e_{s}}v_{e_{s}}^{2}}}\right)^{2}}$.
\\
\\
Now, substituting Eqs. \eqref{x1}, \eqref{x2}, \eqref{x3}, and \eqref{x4} in Eq. (\ref{ascf}) and performing some simple mathematical manipulations, the analytical expression of ASC is derived. Hence, the proof is completed.
\end{proof}

\subsection{Lower Bound of SOP Analysis}
SOP plays a crucial role in evaluating the system's ability to maintain confidentiality against potential eavesdropping \cite{art5}. To ensure secure communication, it is imperative for the instantaneous secrecy capacity ($C_{0}$) to exceed the predetermined threshold secrecy rate ($R_{0}$).

Scenario-$I$:
In this case, it is assumed that the eavesdropper can take information via RF link only. Hence, the lower bound of SOP is defined as \cite[Eq. (21)]{art10}
\begin{align}
\label{S1}
    SOP^{I}=&Pr\{\gamma_{eq}\leq\varphi \gamma_{e_{s}}\}
\nonumber
\\
=&\int_{0}^{\infty}F_{\gamma_{eq}}(\varphi\gamma)f_{\gamma_{e_{s}}}(\gamma)d\gamma,
\end{align}
where $\varphi=2^{R_{0}}$. 

\begin{remark}
   Due to the presence of Meijer's $G$ function in both \eqref{t2} and \eqref{d2}, the derivation of the exact SOP becomes mathematically intractable. Consequently, in this paper, we derive the analytical expression of lower bound of SOP in closed-form.
\end{remark}
\begin{theorem}
The closed form expression of lower bound of $SOP^{I}$ can be written as
\begin{align}
\label{sop1}
     SOP^{I}=&\frac{(\frac{l_{e_{s}}}{\Omega_{e_{s}}})^{\frac{u_{e_{s}}}{2}} }{2\Gamma(u_{e_{s}})v_{e_{s}}^{u_{e_{s}}}}\Biggl[\Re_{1}-\sum_{\eta_{t}=0}^{u_{t}-1}\frac{\varphi^{\frac{\eta_{t}}{2}}}{\eta_{t}!}\left(\sqrt{\frac{l_{t}}{\Omega_{t}v_{t}^{2}}}\right)^{\eta_{t}}\Re_{2}
\nonumber     
\\     +&\sum_{\eta_{t}=0}^{u_{t}-1}\frac{\mathcal{M}_{1}\varphi^{\frac{\eta_{t}}{2}+\frac{u_{d}}{r_{d}}}}{\eta_{t}!\sqrt{r_{d}}(2\pi)^{\frac{r_{d}-1}{2}}}\left(\sqrt{\frac{l_{t}}{\Omega_{t}v_{t}^{2}}}\right)^{\eta_{t}}\Re_{3}\Biggl].
\end{align}
where $\Re_{1}$, $\Re_{2}$, and $\Re_{3}$ are the three integral terms.
\end{theorem}
\begin{proof}
The three integral terms (i.e., $\Re_{1}$, $\Re_{2}$, and $\Re_{3}$) are obtained as follows:
\subsubsection{Derivation of $\Re_{1}$}
$\Re_{1}$ is expressed as
\begin{align}
\Re_{1}=&\int_{0}^{\infty}\gamma^{\frac{u_{e_{s}}}{2}-1}G_{0,1}^{1,0}\left[\sqrt{\frac{l_{e_{s}}}{\Omega_{e_{s}}v_{e_{s}}^{2}}}\gamma^{\frac{1}{2}}\biggl | 
    \begin{array}{c}
     -\\
     0 \\ 
    \end{array}
    \right]d\gamma.
    \nonumber
\end{align}
Now, solving the integral with the help of \cite[Eq. (7.811.4)]{art1}, $\Re_{1}$ is derived finally as
\begin{align}
\label{r1}
\Re_{1}=\frac{2\Gamma(u_{e_{s}})}{\left(\sqrt{\frac{l_{e_{s}}}{\Omega_{e_{s}}v_{e_{s}}^{2}}}\right)^{u_{e_{s}}}}.
\end{align}
\subsubsection{Derivation of $\Re_{2}$}
$\Re_{2}$ is expressed as
\begin{align}
\Re_{2}=&\int_{0}^{\infty}\gamma^{\frac{\eta_{t}}{2}+\frac{u_{e_{s}}}{2}-1}G_{0,1}^{1,0}\left[\sqrt{\frac{l_{t}\varphi}{\Omega_{t}v_{t}^{2}}}\gamma^{\frac{1}{2}}\biggl | 
    \begin{array}{c}
     -\\
     0 \\ 
    \end{array}
    \right]
    \nonumber
    \\
\times&G_{0,1}^{1,0}\left[\sqrt{\frac{l_{e_{s}}}{\Omega_{e_{s}}v_{e_{s}}^{2}}}\gamma^{\frac{1}{2}}\biggl | 
    \begin{array}{c}
     -\\
     0 \\ 
    \end{array}
    \right]d\gamma.
    \nonumber
\end{align}
Using the identity of \cite[Eq. (2.24.1.1)]{art7}, $\Re_{2}$ can be written as
\begin{align}
\label{r2}
\Re_{2}=&2\left(\sqrt{\frac{l_{t}\varphi}{\Omega_{t}v_{t}^{2}}}\right)^{-\eta_{t}-u_{e_{s}}}
\nonumber
\\
\times&G_{1,1}^{1,1}\left[\sqrt{\frac{l_{e_{s}}\Omega_{t}v_{t}^{2}}{l_{t}\varphi\Omega_{e_{s}}v_{e_{s}}^{2}}}\biggl | 
    \begin{array}{c}
     1-\eta_{t}-u_{e_{s}}\\
     0 \\ 
    \end{array}
    \right].
\end{align}
\subsubsection{Derivation of $\Re_{3}$}
$\Re_{3}$ is expressed as
\begin{align}
\Re_{3}=&\int_{0}^{\infty}\gamma^{\frac{\eta_{t}}{2}+\frac{u_{e_{s}}}{2}+\frac{u_{d}}{r_{d}}-1}G_{0,1}^{1,0}\left[\sqrt{\frac{l_{t}\varphi}{\Omega_{t}v_{t}^{2}}}\gamma^{\frac{1}{2}}\biggl | 
    \begin{array}{c}
     -\\
     0 \\ 
    \end{array}
    \right]
    \nonumber
    \\
\times&G_{1,r_{d}+1}^{r_{d},1}\left[\frac{\varphi\mathbb{E}(M_{d})^{r_{d}}}{\Omega_{r_{d}}(v_{d}r_{d})^{r_{d}}}\gamma\biggl | 
    \begin{array}{c}
     1-\frac{u_{d}}{r_{d}} \\
     0,\frac{r_{d}-1}{r_{d}},\frac{-u_{d}}{r_{d}} \\ 
    \end{array}
    \right]
    \nonumber
    \\
\times&G_{0,1}^{1,0}\left[\sqrt{\frac{l_{e_{s}}}{\Omega_{e_{s}}v_{e_{s}}^{2}}}\gamma^{\frac{1}{2}}\biggl | 
    \begin{array}{c}
     -\\
     0 \\ 
    \end{array}
    \right]d\gamma.
    \nonumber
\end{align}
All the Meijer's G terms are transformed into  Fox's H function according to\cite[Eq. (6.2.8)]{art8}. Now, utilizing the identity  \cite[Eq. (2.3)]{art9}, $\Re_{3}$ is derived as
\begin{align}
\label{r3}
\Re_{3}=&\int_{0}^{\infty}\gamma^{\frac{\eta_{t}}{2}+\frac{u_{e_{s}}}{2}+\frac{u_{d}}{r_{d}}-1}H_{0,1}^{1,0}\left[\sqrt{\frac{l_{t}\varphi}{\Omega_{t}v_{t}^{2}}}\gamma^{\frac{1}{2}}\biggl | 
    \begin{array}{c}
     -\\
     (0,1) \\ 
    \end{array}
    \right]
    \nonumber
    \\
\times&H_{1,r_{d}+1}^{r_{d},1}\left[\frac{\varphi\mathbb{E}(M_{d})^{r_{d}}}{\Omega_{r_{d}}(v_{d}r_{d})^{r_{d}}}\gamma\biggl | 
    \begin{array}{c}
     (1-\frac{u_{d}}{r_{d}},1) \\
     (0,1),(\frac{r_{d}-1}{r_{d}},1),(\frac{-u_{d}}{r_{d}},1) \\ 
    \end{array}
    \right]
\nonumber
\\
\times&H_{0,1}^{1,0}\left[\sqrt{\frac{l_{e_{s}}}{\Omega_{e_{s}}v_{e_{s}}^{2}}}\gamma^{\frac{1}{2}}\biggl | 
    \begin{array}{c}
     -\\
     (0,1)\\ 
    \end{array}
    \right]d\gamma
\nonumber
\\
=&2\varpi_{1}H_{1,0;0,1;1,r_{d}+1}^{1,0;1,0;r_{d},1}\left[
    \begin{array}{c}
     \varpi_{2}\\
     1;- \\ 
    \end{array}\biggl |    \begin{array}{c}
     -\\
     (0,1)
    \end{array}\biggl | \begin{array}{c}
     \varpi_{3}\\
     \varpi_{4}
    \end{array}\biggl | \begin{array}{c}
     \varpi_{5}, \varpi_{6} \\ 
    \end{array}  
    \right],
\end{align}
where 
\setcounter{eqnback}{\value{equation}}
\setcounter{equation}{33}
\begin{figure*}[!b]
\hrulefill
\begin{align}
\label{sop2}SOP^{II}=&1+\sum_{\eta_{t}=0}^{u_{t}-1}\frac{1}{\eta_{t}!}\left(\sqrt{\frac{l_{t}(\varphi-1)}{\Omega_{t}v_{t}^{2}}}\right)^{\eta_{t}}G_{0,1}^{1,0}\left[\sqrt{\frac{l_{t}(\varphi-1)}{\Omega_{t}v_{t}^{2}}}\biggl | 
    \begin{array}{c}
     -\\
     0 \\ 
    \end{array}
    \right]\Biggl\{\frac{\varphi^{-\frac{u_{e_{r}}}{r_{e_{r}}}}\Omega_{r_{e_{r}}}^{-\frac{u_{e_{r}}}{r_{e_{r}}}}\Omega_{r_{d}}^{-\frac{u_{d}}{r_{d}}}\mathbb{E}(M_{e_{r}})^{u_{e_{r}}}\mathbb{E}(M_{d})^{u_{d}}}{\sqrt{r_{e_{r}}}\sqrt{r_{d}}\Gamma(u_{e_{r}})\Gamma(u_{d})v_{e_{r}}^{u_{e_{r}}}v_{d}^{u_{d}}(2\pi)^{\frac{r_{d}-1}{2}}\left(2\pi\right)^{\frac{r_{e_{r}}-1}{2}}}
    \nonumber
    \\
\times&\Biggl(\frac{\mathbb{E}(M_{d})^{r_{d}}}{\Omega_{r_{d}}(v_{d}r_{d})^{r_{d}}}\Biggl)^{-\left(\frac{u_{d}}{r_{d}}+\frac{u_{e_{r}}}{r_{e_{r}}}\right)}G_{r_{e_{r}}+1,r_{e_{r}}+1}^{r_{e_{r}}+1,r_{d}}\left[\frac{\mathbb{E}(M_{e_{r}})^{{r_{e_{r}}}}(v_{d}r_{d})^{r_{d}}\Omega_{r_{d}}}{\mathbb{E}(M_{d})^{r_{d}}(v_{e_{r}}r_{e_{r}})^{r_{e_{r}}}\Omega_{e_{r}}\phi }\biggl | 
    \begin{array}{c}
     \varpi_{7},\varpi_{7}-\mathcal{S}_{r_{d}},\varpi_{7}-\mathcal{S}_{r_{d}+1}\\
     0,\frac{u_{e_{r}}}{r_{e_{r}}},\frac{r_{e_{r}}-1}{r_{e_{r}}} \\ 
    \end{array}
    \right]-1\Biggl\},
\end{align} 
\end{figure*}
\setcounter{eqnback}{\value{equation}}
\setcounter{equation}{34}
\begin{figure*}[!b]
\hrulefill
\begin{align}
\label{sop2a}SOP^{II}_{\infty}=&1+\sum_{\eta_{t}=0}^{u_{t}-1}\frac{1}{\eta_{t}!}\left(\sqrt{\frac{l_{t}(\varphi-1)}{\Omega_{t}v_{t}^{2}}}\right)^{\eta_{t}}
\Biggl\{\sum_{i=1}^{r_{d}}\frac{\varphi^{-\frac{u_{e_{r}}}{r_{e_{r}}}}\Omega_{r_{e_{r}}}^{-\frac{u_{e_{r}}}{r_{e_{r}}}}\Omega_{r_{d}}^{-\frac{u_{d}}{r_{d}}}\mathbb{E}(M_{e_{r}})^{u_{e_{r}}}\mathbb{E}(M_{d})^{u_{d}}}{\sqrt{r_{e_{r}}}\sqrt{r_{d}}\Gamma(u_{e_{r}})\Gamma(u_{d})v_{e_{r}}^{u_{e_{r}}}v_{d}^{u_{d}}(2\pi)^{\frac{r_{d}-1}{2}}\left(2\pi\right)^{\frac{r_{e_{r}}-1}{2}}}
    \nonumber
    \\
\times&\Biggl(\frac{\mathbb{E}(M_{d})^{r_{d}}}{\Omega_{r_{d}}(v_{d}r_{d})^{r_{d}}}\Biggl)^{-\left(\frac{u_{d}}{r_{d}}+\frac{u_{e_{r}}}{r_{e_{r}}}\right)}
\Biggl(\frac{\mathbb{E}(M_{e_{r}})^{{r_{e_{r}}}}(v_{d}r_{d})^{r_{d}}\Omega_{r_{d}}}{\mathbb{E}(M_{d})^{r_{d}}(v_{e_{r}}r_{e_{r}})^{r_{e_{r}}}\Omega_{e_{r}}\phi}\Biggl)^{H_{i}-1}\frac{\prod_{\kappa=1;\kappa\neq i}^{r_{d}}\Gamma(H_{i}-H_{\kappa})\prod_{\kappa=1}^{r_{e_{r}}+1}\Gamma(1+G_{\kappa}-H_{i})}{\prod_{\kappa=r_{d}+1}^{r_{e_{r}+1}}\Gamma(1+H_{\kappa}-H_{i})}\Biggl\},
\end{align}
\end{figure*}
$\varpi_{1}=\left(\sqrt{\frac{l_{t}\varphi}{\Omega_{t}v_{t}^{2}}}\right)^{-\left(\eta_{t}+u_{e_{s}}+\frac{2u_{d}}{r_{d}}\right)}$,
$\varpi_{2}=(1-\eta_{t}-u_{e_{s}}-\frac{2u_{d}}{r_{d}};1,2)$,
$\varpi_{3}=(1-\frac{u_{d}}{r_{d}},1)$,
$\varpi_{4}=\{ (0,1),(\frac{r_{d}-1}{r_{d}},1),(\frac{-u_{d}}{r_{d}},1)\}$,
$\varpi_{5}=\sqrt{\frac{l_{e_{s}}\Omega_{t}v_{t}^{2}}{l_{t}\varphi\Omega_{e_{s}}v_{e_{s}}^{2}}}$, and
$\varpi_{6}=\frac{\Omega_{t}v_{t}^{2}\mathbb{E}(M_{d})^{r_{d}}}{l_{t}\Omega_{r_{d}}(v_{d}r_{d})^{r_{d}}}$.
\\
Now, substituting Eqs. \eqref{r1}, \eqref{r2}, and \eqref{r3} in Eq. \eqref{sop1} and performing simple arithmetical manipulations, the analytical expression of SOP (Scenario-$I$) is derived. Hence, the proof is completed.
\end{proof}

\begin{remark}
  \label{re1} 
Assuming that $l_{t}=l_{e_{s}}$, $\Omega_{t}=\Omega_{e_{s}}$, $u_{t}=u_{e_{s}}$, $v_{t}=v_{e_{s}}$, and $N_{t}=N_{e_{s}}$, $\mathcal{E}1$ can intercept confidential information via $\mathcal{I}_{\mathcal{R}}$. Similarly, $\mathcal{E}2$ can intercept data transmission via $\mathcal{I}_{\mathcal{D}}$ under the assumption that $\Omega_{r_{d}}=\Omega_{r_{e_{r}}}$, $u_{d}=u_{e_{r}}$, $v_{d}=v_{e_{r}}$, and $N_{d}=N_{e_{r}}$. 
\end{remark}

\noindent  \textit{Asymptotic Analysis:}
By substituting \eqref{feqasy} and \eqref{t1} into \eqref{S1} and using the formula of \cite[Eq. (7.811.4)]{art1}, the asymptotic form of $SOP^{I}$ is obtained  as
\setcounter{eqnback}{\value{equation}}
\setcounter{equation}{30}
\begin{align}
    \label{sopa}    SOP^{I}_{\infty}=&\frac{(\frac{l_{e_{s}}}{\Omega_{e_{s}}})^{\frac{u_{e_{s}}}{2}} }{\Gamma(u_{e_{s}})v_{e_{s}}^{u_{e_{s}}}}\Biggl[\frac{1}{u_{t}!}\Biggl(\sqrt{\frac{l_{t}}{\Omega_{t}v_{t}^{2}}}\Biggl)^{u_{t}}  \Biggl(\sqrt{\frac{l_{e_{s}}}{\Omega_{e_{s}}v_{e_{s}}^{2}}}\Biggl)^{-u_{t}-u_{e_{s}}}
\nonumber
\\
\times&\Gamma(u_{t}+u_{e_{s}})+\sum_{i=1}^{r_{d}}\Biggl(\sqrt{\frac{l_{e_{s}}}{\Omega_{e_{s}}v_{e_{s}}^{2}}}\Biggl)^{-(u_{e_{s}}+\frac{2u_{d}}{r_{d}}+2p_{i}-2)}
\nonumber
\\
\times&\frac{\mathcal{M}_{1}}{\sqrt{r_{d}}(2\pi)^{\frac{r_{d}-1}{2}}}
\frac{\prod_{\kappa=1;\kappa\neq i}^{r_{d}}(P_{i}-P_{\kappa})\Gamma(1+\frac{u_{d}}{r_{d}}-P_{i})}{\Gamma(1+P_{r_{d}+1}-P_{i})}
\nonumber
\\
\times&\Biggl(\frac{\mathbb{E}(M_{d})^{r_{d}}}{\Omega_{r_{d}}(v_{d}r_{d})^{r_{d}}}\Biggl)^{P_{i}-1}\Gamma\biggl(u_{e_{s}}+\frac{2u_{d}}{r_{d}}+2p_{i}-2\biggl)
\Biggl].
\end{align}

Scenario-$II$: In this case, it is assumed that eavesdroppers are situated underwater, intercepting the optical link with the intention of obtaining confidential information. For this scenario, the SOP can be expressed mathematically as \cite[Eq. (13)]{art11}
\begin{align}
SOP^{II}=&Pr\{C_{0}\leq R_{0}\}=Pr\{\gamma_{eq}\leq \varphi \gamma_{e_{r}}+\varphi-1\}
\nonumber
\\
=&\int_{0}^{\infty}\int_{\varphi\gamma+\varphi-1}^{\infty}f_{\gamma_{eq}}(\gamma)f_{\gamma_{e_{r}}}(\gamma)d\gamma_{eq}d\gamma
\nonumber
\\
=&\int_{0}^{\infty}F_{\gamma_{d}}(\varphi\gamma+\varphi-1)f_{\gamma_{e_{r}}}(\gamma)d\gamma
\nonumber
\\
\times& \left(1- F_{\gamma_{t}}(\varphi-1)\right)+F_{\gamma_{t}}(\varphi-1).
\end{align}
Since it is difficult to determine the exact closed-form formula of SOP through mathematical calculations, we provide the analytical expression of SOP considering the lower bound of SOP as \cite[Eq. (14)]{art11}
\begin{align}
\label{sop2aa}
    SOP^{II} \ge SOP^{L}=&Pr\{\gamma_{eq}\leq\varphi \gamma_{e_{r}}\}
\nonumber
\\
=&\int_{0}^{\infty}F_{\gamma_{eq}}(\varphi\gamma)f_{\gamma_{e_{s}}}(\gamma)d\gamma\
\nonumber
\\
=&\int_{0}^{\infty}F_{\gamma_{d}}(\varphi\gamma)f_{\gamma_{e_{r}}}(\gamma)d\gamma
\nonumber
\\
\times& \left(1- F_{\gamma_{t}}(\varphi-1)\right)+F_{\gamma_{t}}(\varphi-1).
\end{align}
\begin{theorem}
The analytical expression of $SOP^{II}$ can be demonstrated as \eqref{sop2}, where $\varpi_{7}=1-\frac{u_{d}}{r_{d}}-\frac{u_{e_{r}}}{r_{e_{r}}}$ and
$\mathcal{S}=\begin{bmatrix}
    0 & \frac{-u_{e_{r}}}{r_{e_{r}}} & \frac{r_{e_{r}}-1}{r_{e_{r}}} 
\end{bmatrix}$. 
\end{theorem}
\begin{proof}
By substituting Eqs. \eqref{t2}, \eqref{d1}, and \eqref{d2} into Eq. \eqref{sop2aa} and subsequently applying the identity of \cite[Eq. (2.24.1.1)]{art7}, the closed form expression of $SOP^{II}$ is obtained. Hence, the proof is completed.
\end{proof}
\noindent 
\\
\textit{Asymptotic Analysis:}
Similar to Scenario-$I$, the asymptotic $SOP^{II}$ are analyzed for high SNR region as shown in \eqref{sop2a} by using the similar mathematical formula as utilized in \eqref{f2}, where $H=\begin{bmatrix}
    \varpi_{7} & \varpi_{7}-\mathcal{S}_{r_{d}} & \varpi_{7}-\mathcal{S}_{r_{d}+1} 
\end{bmatrix}$ and $G=\begin{bmatrix}
    0 & \frac{u_{e_{r}}}{r_{e_{r}}} & \frac{r_{e_{r}}-1}{r_{e_{r}}} 
\end{bmatrix}$ are two constants.

Scenario-$III$: In this case, we consider the existence of two distinct eavesdroppers, each with the ability to intercept identical information from two separate channels, such as RF and optical links. Mathematically, $SOP^{III}$ can be expressed as 
\setcounter{eqnback}{\value{equation}}
\setcounter{equation}{35}
\begin{align}
\label{sop3}
    SOP^{III}=1-(SOP_{u}\times SOP_{r}),
\end{align}
where 
\begin{align}
\label{so1}
    SOP_{u}=&1-\int_{0}^{\infty}F_{\gamma_{t}}(\varphi\gamma)f_{\gamma_{e_{s}}}(\gamma)d\gamma.
\\
\label{so2}     
 SOP_{r}=&1-\int_{0}^{\infty}F_{\gamma_{d}}(\varphi\gamma)f_{\gamma_{e_{r}}}(\gamma)d\gamma.
\end{align}

\begin{theorem}
Assuming simultaneous eavesdropping attack, the closed-form expression of SOP is obtained as \eqref{sopp3}.
\end{theorem}
\begin{proof}
Substituting Eqs. \eqref{t1} and \eqref{t2} into Eq. \eqref{so1} and using the identities of \cite[Eq. (7.811.4)]{art1} and \cite[Eq. (2.24.1.1)]{art7}, $ SOP_{u}$ is derived as
\setcounter{eqnback}{\value{equation}}
\setcounter{equation}{39}
\begin{align}
    SOP_{u}=&\sum_{\eta_{t}=0}^{u_{t}-1}\frac{v_{e_{s}}^{-u_{e_{s}}}}{\eta_{t}!\Gamma(u_{e_{s}})}\left(\sqrt{\frac{l_{t}}{\Omega_{t}v_{t}^{2}}}\right)^{\eta_{t}}\left(\sqrt{\frac{l_{t}\varphi}{\Omega_{t}v_{t}^{2}}}\right)^{-\eta_{t}-u_{e_{s}}}
\nonumber
\\
\times&\biggl(\frac{l_{e_{s}}}{\Omega_{e_{s}}}\biggl)^{\frac{u_{e_{s}}}{2}}G_{1,1}^{1,1}\left[\sqrt{\frac{l_{e_{s}}\Omega_{t}v_{t}^{2}}{l_{t}\phi\Omega_{e_{s}}v_{e_{s}}^{2}}}\biggl | 
    \begin{array}{c}
     1-\eta_{t}-u_{e_{s}}\\
     0 \\ 
    \end{array}
    \right].
\end{align}
Now, substituting Eqs. \eqref{d1} and \eqref{d2} into Eq. \eqref{so2} and utilizing the identity of \cite[Eq. (2.24.1.1)]{art7}, $SOP_{r}$ is obtained as shown in \eqref{o2}.
\end{proof}
\setcounter{eqnback}{\value{equation}}
\setcounter{equation}{38}
\begin{figure*}[!b]
\hrulefill
\begin{align}
    \label{sopp3}SOP^{III}=&1-\Biggl\{\sum_{\eta_{t}=0}^{u_{t}-1}\frac{1}{\eta_{t}!\Gamma(u_{e_{s}})v_{e_{s}}^{u_{e_{s}}}}\left(\sqrt{\frac{l_{t}}{\Omega_{t}v_{t}^{2}}}\right)^{\eta_{t}}\left(\sqrt{\frac{l_{t}\varphi}{\Omega_{t}v_{t}^{2}}}\right)^{-\eta_{t}-u_{e_{s}}}\biggl(\frac{l_{e_{s}}}{\Omega_{e_{s}}}\biggl)^{\frac{u_{e_{s}}}{2}}G_{1,1}^{1,1}\left[\sqrt{\frac{l_{e_{s}}\Omega_{t}v_{t}^{2}}{l_{t}\phi\Omega_{e_{s}}v_{e_{s}}^{2}}}\biggl | 
    \begin{array}{c}
     1-\eta_{t}-u_{e_{s}}\\
     0 \\ 
    \end{array}
    \right]\Biggl\}
\nonumber
\\
\times&\Biggl\{1-\frac{\varphi^{-\frac{u_{e_{r}}}{r_{e_{r}}}}\Omega_{r_{e_{r}}}^{-\frac{u_{e_{r}}}{r_{e_{r}}}}\Omega_{r_{d}}^{-\frac{u_{d}}{r_{d}}}\mathbb{E}(M_{e_{r}})^{u_{e_{r}}}\mathbb{E}(M_{d})^{u_{d}}}{\sqrt{r_{e_{r}}}\sqrt{r_{d}}\Gamma(u_{e_{r}})\Gamma(u_{d})v_{e_{r}}^{u_{e_{r}}}v_{d}^{u_{d}}(2\pi)^{\frac{r_{d}-1}{2}}\left(2\pi\right)^{\frac{r_{e_{r}}-1}{2}}}\Biggl(\frac{\mathbb{E}(M_{d})^{r_{d}}}{\Omega_{r_{d}}(v_{d}r_{d})^{r_{d}}}\Biggl)^{-\left(\frac{u_{d}}{r_{d}}+\frac{u_{e_{r}}}{r_{e_{r}}}\right)}
    \nonumber
    \\
\times&G_{r_{e_{r}}+1,r_{e_{r}}+1}^{r_{e_{r}}+1,r_{d}}\left[\frac{\mathbb{E}(M_{e_{r}})^{{r_{e_{r}}}}(v_{d}r_{d})^{r_{d}}\Omega_{r_{d}}}{\mathbb{E}(M_{d})^{r_{d}}(v_{e_{r}}r_{e_{r}})^{r_{e_{r}}}\Omega_{e_{r}}\phi }\biggl | 
    \begin{array}{c}
     \varpi_{7},\varpi_{7}-\mathcal{S}_{r},\varpi_{7}-\mathcal{S}_{r+1}\\
     0,\frac{u_{e_{r}}}{r_{e_{r}}},\frac{r_{e_{r}}-1}{r_{e_{r}}} \\ 
    \end{array}
    \right]\Biggl\}.
\end{align} 
\end{figure*}
\setcounter{eqnback}{\value{equation}}
\setcounter{equation}{40}
\begin{figure*}[!b]
\hrulefill
\begin{align}
    \label{o2}SOP_{r}=&1-\frac{\varphi^{-\frac{u_{e_{r}}}{r_{e_{r}}}}\Omega_{r_{e_{r}}}^{-\frac{u_{e_{r}}}{r_{e_{r}}}}\Omega_{r_{d}}^{-\frac{u_{d}}{r_{d}}}\mathbb{E}(M_{e_{r}})^{u_{e_{r}}}\mathbb{E}(M_{d})^{u_{d}}}{\sqrt{r_{e_{r}}}\sqrt{r_{d}}\Gamma(u_{e_{r}})\Gamma(u_{d})v_{e_{r}}^{u_{e_{r}}}v_{d}^{u_{d}}(2\pi)^{\frac{r_{d}-1}{2}}\left(2\pi\right)^{\frac{r_{e_{r}}-1}{2}}}\Biggl(\frac{\mathbb{E}(M_{d})^{r_{d}}}{\Omega_{r_{d}}(v_{d}r_{d})^{r_{d}}}\Biggl)^{-\left(\frac{u_{d}}{r_{d}}+\frac{u_{e_{r}}}{r_{e_{r}}}\right)}
    \nonumber
    \\
\times&G_{r_{e_{r}}+1,r_{e_{r}}+1}^{r_{e_{r}}+1,r_{d}}\left[\frac{\mathbb{E}(M_{e_{r}})^{{r_{e_{r}}}}(v_{d}r_{d})^{r_{d}}\Omega_{r_{d}}}{\mathbb{E}(M_{d})^{r_{d}}(v_{e_{r}}r_{e_{r}})^{r_{e_{r}}}\Omega_{e_{r}}\phi }\biggl | 
    \begin{array}{c}
     \varpi_{7},\varpi_{7}-\mathcal{S}_{r},\varpi_{7}-\mathcal{S}_{r+1}\\
     0,\frac{u_{e_{r}}}{r_{e_{r}}},\frac{r_{e_{r}}-1}{r_{e_{r}}} \\ 
    \end{array}
    \right].
\end{align} 
\end{figure*}
\setcounter{eqnback}{\value{equation}}
\setcounter{equation}{41}
\begin{figure*}
\begin{align}
    \label{sopp3a}SOP^{III}_{\infty}=&1-\Biggl\{\sum_{\eta_{t}=0}^{u_{t}-1}\frac{\Gamma(\eta_{t}+u_{e_{s}})}{\eta_{t}!\Gamma(u_{e_{s}})v_{e_{s}}^{u_{e_{s}}}}\left(\sqrt{\frac{l_{t}}{\Omega_{t}v_{t}^{2}}}\right)^{\eta_{t}}\left(\sqrt{\frac{l_{t}\varphi}{\Omega_{t}v_{t}^{2}}}\right)^{-\eta_{t}-u_{e_{s}}}\biggl(\frac{l_{e_{s}}}{\Omega_{e_{s}}}\biggl)^{\frac{u_{e_{s}}}{2}}\Bigg(\sqrt{\frac{l_{e_{s}}\Omega_{t}v_{t}^{2}}{l_{t}\phi\Omega_{e_{s}}v_{e_{s}}^{2}}}\Biggl)^{\eta_{t}-u_{e_{s}}}\Biggl\}
\nonumber
\\
\times&\Biggl\{1-\sum_{i=1}^{r_{d}}\frac{\varphi^{-\frac{u_{e_{r}}}{r_{e_{r}}}}\Omega_{r_{e_{r}}}^{-\frac{u_{e_{r}}}{r_{e_{r}}}}\Omega_{r_{d}}^{-\frac{u_{d}}{r_{d}}}\mathbb{E}(M_{e_{r}})^{u_{e_{r}}}\mathbb{E}(M_{d})^{u_{d}}}{\sqrt{r_{e_{r}}}\sqrt{r_{d}}\Gamma(u_{e_{r}})\Gamma(u_{d})v_{e_{r}}^{u_{e_{r}}}v_{d}^{u_{d}}(2\pi)^{\frac{r_{d}-1}{2}}\left(2\pi\right)^{\frac{r_{e_{r}}-1}{2}}}\Biggl(\frac{\mathbb{E}(M_{d})^{r_{d}}}{\Omega_{r_{d}}(v_{d}r_{d})^{r_{d}}}\Biggl)^{-\left(\frac{u_{d}}{r_{d}}+\frac{u_{e_{r}}}{r_{e_{r}}}\right)}
    \nonumber
    \\
\times&\Biggl(\frac{\mathbb{E}(M_{e_{r}})^{{r_{e_{r}}}}(v_{d}r_{d})^{r_{d}}\Omega_{r_{d}}}{\mathbb{E}(M_{d})^{r_{d}}(v_{e_{r}}r_{e_{r}})^{r_{e_{r}}}\Omega_{e_{r}}\phi}\Biggl)^{H_{i}-1}\frac{\prod_{\kappa=1;\kappa\neq i}^{r_{d}}\Gamma(H_{i}-H_{\kappa})\prod_{\kappa=1}^{r_{e_{r}}+1}\Gamma(1+G_{\kappa}-H_{i})}{\prod_{\kappa=r_{d}+1}^{r_{e_{r}+1}}\Gamma(1+H_{\kappa}-H_{i})}\Biggl\}.
\end{align} 
\hrulefill
\end{figure*}
\\
\noindent 
\textit{Asymptotic analysis:}
Using the same identity as used in Eqs. \eqref{f2} and \eqref{sop2a}, the asymptotic representation of $SOP^{III}$ is obtained finally as Eq. \eqref{sopp3a}.

\subsection{SPSC Analysis}
The significance of the probability of SPSC lies in its role as a crucial performance metric, ensuring the persistence of a data stream only when the secrecy capacity remains positive. Mathematically, the probability of SPSC can be written as \cite[Eq. (28)]{art12}
\begin{align}
\label{SPSC}
   SPSC=&Pr\{C_{0}>0\}
   \nonumber
   \\
=&1-SOP|_{R_{0}=0}
\end{align}

\begin{theorem}
The closed-form analytical expressions of SPSC can be expressed as
\begin{align}
\label{spsc1}
SPSC^{I}=&1-SOP^{I}|_{R_{0}=0} (Scenario-I)
\\
\label{spsc2}
SPSC^{II}=&1-SOP^{II}|_{R_{0}=0} (Scenario-II)
\\
\label{spsc3}
SPSC^{III}=&1-SOP^{III}|_{R_{0}=0} (Scenario-III)
\end{align}
\end{theorem}
\begin{proof}
By substituting Eqs. \eqref{sop1}, \eqref{sop2}, and \eqref{sop3} into Eq. \eqref{SPSC}, the probability of SPSC is obtained. Hence, the proof is completed.
\end{proof}
\subsection{EST Analysis}
EST serves as a performance metric that combines and emphasizes the dependability of the tapping channel along with the constraints of indemnity. It is the mean rate that ensures secure data transmission from the source to the destination without any interruption. According to \cite[Eq. (5)]{art13}, the EST can be written mathematically as 
\begin{align}
\label{EST}
    EST=R_{0}(1-SOP).
\end{align}
\begin{theorem}
The closed-form expressions of EST can be demonstrated as
\begin{align}
\label{est1}
EST^{I}=&R_{0}(1-SOP^{I}) (Scenario-I)
\\
\label{est2}
EST^{II}=&R_{0}(1-SOP^{II}) (Scenario-II)
\\
\label{est3}
EST^{III}=&R_{0}(1-SOP^{III}) (Scenario-III)
\end{align}
\end{theorem}
\begin{proof}
By substituting Eqs. \eqref{sop1}, \eqref{sop2} and \eqref{sop3} into Eq. \eqref{EST}, the EST is derived. Hence, the proof is completed.
\end{proof}



\section{NUMERICAL RESULTS}
\label{NR}
This section provides graphical representations of numerical results derived from the analytical expressions outlined in Section \ref{spa}. Our focus is specifically on secrecy metrics such as ASC, SOP, the probability of SPSC, and EST analysis, as shown in
\eqref{ascf}, \eqref{sop1}, \eqref{sop2}, \eqref{sopp3}, \eqref{spsc1}, \eqref{spsc2}, \eqref{spsc3}, \eqref{est1}, \eqref{est2}, and \eqref{est3}. Utilizing these closed-form expressions, we investigate the influence of various system parameters ($N_{t}$, $N_{d}$, $N_{e_{s}}$, $N_{e_{r}}$, $\alpha_{t,s}$, $\alpha_{t,r}$, $\alpha_{e_{s},s}$, $\alpha_{e_{s},r}$, $\mu_{t,s}$, $\mu_{t,r}$, $\mu_{e_{s},s}$, $\mu_{e_{s},r}$, $\Omega_{t}$, $\Omega_{e_{s}}$, $\Omega_{r_{d}}$, $\Omega_{r_{e_{r}}}$, $\phi_{t,r}$, $\xi_{d}$, and $\xi_{e_{r}}$) on the secrecy performance of the proposed RF-UOWC mixed network. Additionally, we underscore our analytical insights into three distinct eavesdropping scenarios: RF eavesdropping, UOWC eavesdropping, and the simultaneous presence of both RF and UOWC eavesdropping. According to \cite{d1,10413214}, we make the following assumptions: $N_{t}=N_{d}=N_{e_{s}}=N_{e_{r}}=1$, $\alpha_{t,s}=\alpha_{e_{s},s}=\alpha_{t,r}=\alpha_{e_{s},r}=2$, $\mu_{t,s}=\mu_{t,r}=\mu_{e_{s},s}=\mu_{e_{s},r}=4$, $\phi_{t,s}=\phi_{t,r}=1$, $\xi_{d}=\xi_{e_{r}}=1$, $R_{0}=0.01$ bits/sec/Hz, $r=(1,2)$, $\Omega_{e_{s}}= 25$ dB, and $\Omega_{r_{d}}=\Omega_{r_{e_{r}}}= 20$ dB, unless specified otherwise. To validate the analytical conclusions, we conduct Monte Carlo (MC) simulations by generating $\alpha-\mu$ and mEGG random variables using MATLAB. This simulation entails averaging $10^{6}$ channel realizations to derive values for each secrecy indicator. The accuracy of the derived expressions is confirmed through a close alignment between the analytical findings and simulation results. Furthermore, the asymptotic analysis is performed, demonstrating a robust agreement in the high SNR regime with the analytical results.

\subsection{Impact of RIS reflecting elements}
To assess the interplay between security and RIS technology within the proposed RF-UOWC mixed model, the impacts of the RIS reflecting elements (i.e., $N_{t}$, $N_{d}$, $N_{e_{s}}$, and $N_{e_{r}}$) on the secrecy performance are examined in Figs. \ref{p12}-\ref{p11}.
\begin{figure}[!ht]
\vspace{0mm}
\centerline{\includegraphics[width=0.35\textwidth,angle =0]{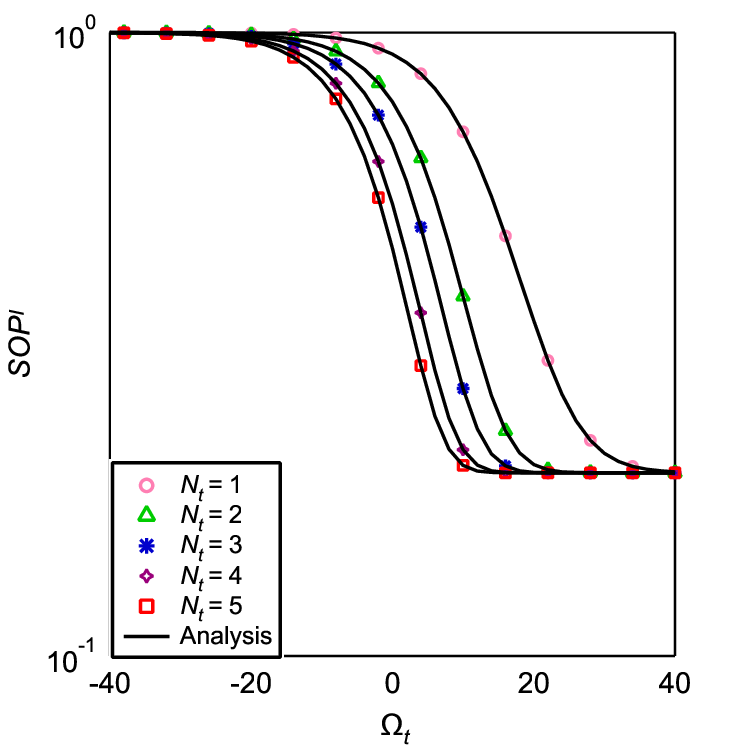}}
       \vspace{0mm}
   \caption{The $SOP^{I}$ versus $\Omega_{t}$ for chosen values of $N_{t}$.}
   \label{p12}
\end{figure}
Notably, it can be seen from Fig. \ref{p12} that SOP performance is better when the value of $N_t$ is higher. This is expected because, with more reflecting elements, the RIS can manipulate radio waves more precisely, focusing the signal towards the relay, $\mathcal{R}$. This leads to a stronger signal and potentially extends the coverage area.
\begin{figure}[!ht]
\vspace{0mm}
\centerline{\includegraphics[width=0.35\textwidth,angle =0]{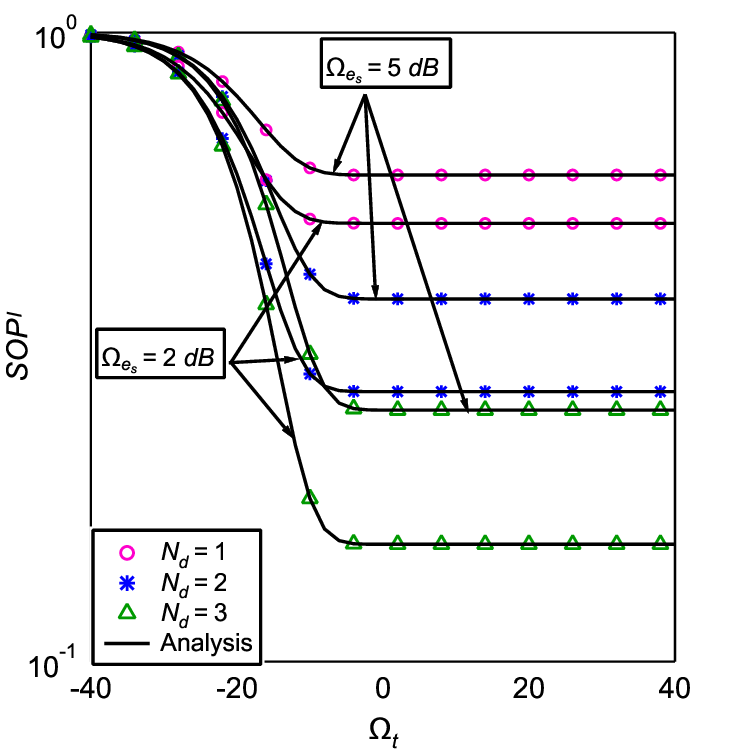}}
       \vspace{0mm}
   \caption{The $SOP^{I}$ versus $\Omega_{t}$ for chosen values of $N_{d}$ and $\Omega_{e_{s}}$.}
   \label{p6}
\end{figure}

In Fig. \ref{p6}, a similar positive impact is observed with an increased $N_{d}$, causing a significant reduction in $SOP^{I}$ when $N_{d}$ is raised from $1$ to $3$. Hence, this improvement plays a pivotal role in reducing the potential risk of eavesdropping, as a greater quantity of reflecting elements enables more precise and focused beam-forming. Consequently, this leads to an increased SNR at the receiver end, effectively mitigating the adverse effects of fading.

\begin{figure}[!ht]
\vspace{0mm}
\centerline{\includegraphics[width=0.35\textwidth,angle =0]{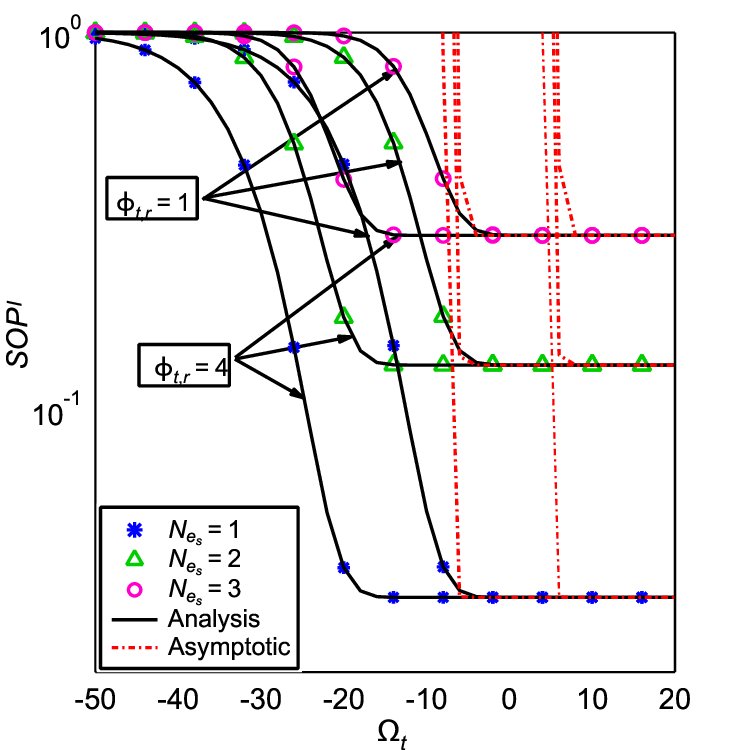}}
       \vspace{0mm}
   \caption{The $SOP^{I}$ versus $\Omega_{t}$ for chosen values of $N_{e_{s}}$ and $\phi_{t,r}$.}
   \label{p5}
\end{figure}
Despite the fact that increasing the number of reflecting elements of an RIS generally benefits secure communication,  Figs. \ref{p5}-\ref{p11} reveal a counter-intuitive effect. Here, an increase in the number of reflecting elements (i.e., $N_{e_{s}}$ and $N_{e_{r}}$) associated with eavesdroppers RIS (i.e., $\mathcal{I}_{{\mathcal{E}1}}$ and $\mathcal{I}_{{\mathcal{E}2}}$) corresponds to a heightened vulnerability to eavesdropping.
\begin{figure}[!ht]
\vspace{0mm}
\centerline{\includegraphics[width=0.35\textwidth,angle =0]{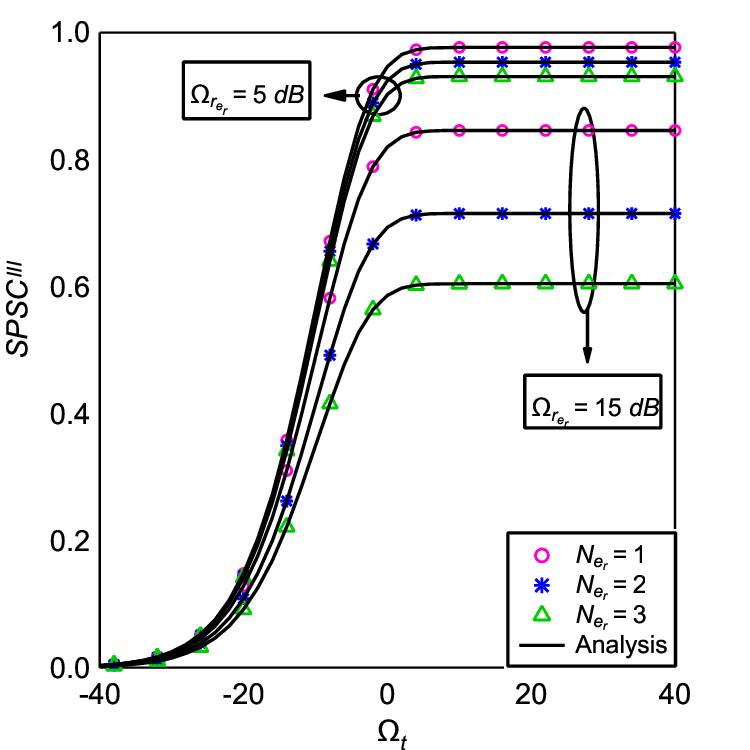}}
       \vspace{0mm}
   \caption{The $SPSC^{III}$ versus $\Omega_{t}$ for chosen values of $N_{e_{r}}$ and $\Omega_{r_{e_{r}}}$.}
   \label{p11}
\end{figure}
For instance, $SOP^{I}$ increases when $N_{e_{s}}$ is raised from $1$ to $3$. This is due to the fact that with more elements in $\mathcal{I}_{{\mathcal{E}1}}$, the eavesdropper can potentially capture a stronger or more focused version of the signal reflected by the RIS. This makes it easier for them to eavesdrop on the communication. Therefore, Fig. \ref{p11} aligns with these observations, demonstrating a decrease in the SPSC (Scenario-$III$) as the number of $N_{e_{r}}$ increases. This vulnerability arises because a larger $N_{e_{r}}$ facilitates eavesdropping by enabling the capture of a stronger or more focused signal reflection.

\subsection{Impact of RF fading parameters}Wireless communication, owing to its broadcast nature, is prone to eavesdropping and interception. Thus, analyzing the effects of RF fading parameters on secrecy performance is instrumental in alleviating security vulnerabilities. Consequently, Figs. \ref{p4} - \ref{p13} delve into the influence of fading parameters specifically concerning the $\mathcal{S}-\mathcal{I}_{{{\mathcal{R}}}}-\mathcal{R}$ link.
\begin{figure}[!ht]
\vspace{0mm}
\centerline{\includegraphics[width=0.35\textwidth,angle =0]{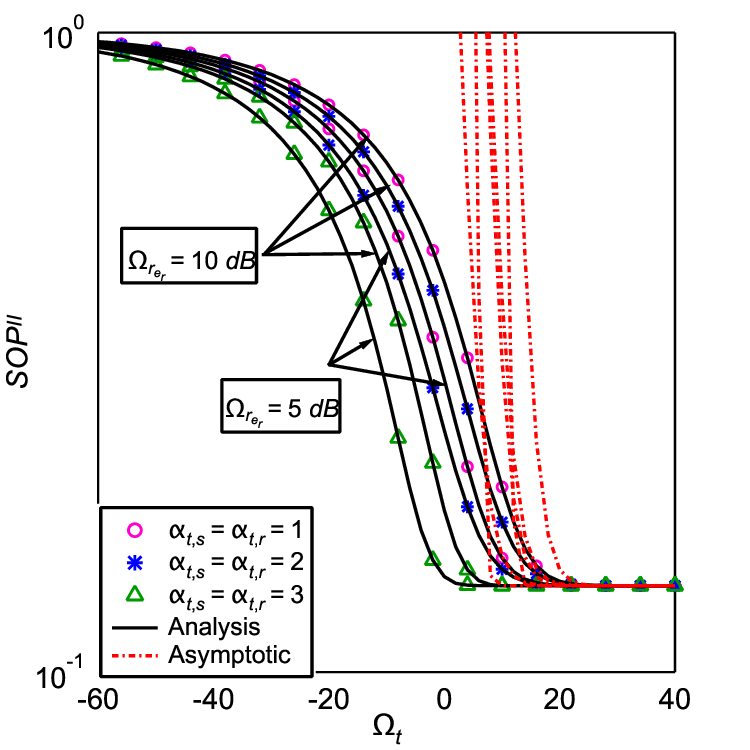}}
       \vspace{0mm}
   \caption{The $SOP^{II}$ versus $\Omega_{t}$ for different chosen values of $\alpha_{t,s}$, $\alpha_{t,r}$, and $\Omega_{r_{e_{r}}}$.}
   \label{p4}
\end{figure}
Observations from Fig. \ref{p4} reveal that increasing both $\alpha_{t,s}$ and $\alpha_{t,r}$ from $1$ to $3$ results in a decrease in $SOP^{II}$ (Scenario-$II$). Consequently, these findings suggest that increasing the fading parameters ($\alpha_{t,s}$ and $\alpha_{t,r}$) leads to a reduction in fading severity, thus improving the reliability of the signal received at $\mathcal{R}$. A similar positive outcome is observed in Fig. \ref{p10}, where the value of $SPSC^{II}$ increases when $\mu_{t,s}$ and $\mu_{t,r}$ rises from $1$ to $10$.
\begin{figure}[!ht]
\vspace{0mm}
\centerline{\includegraphics[width=0.35\textwidth,angle =0]{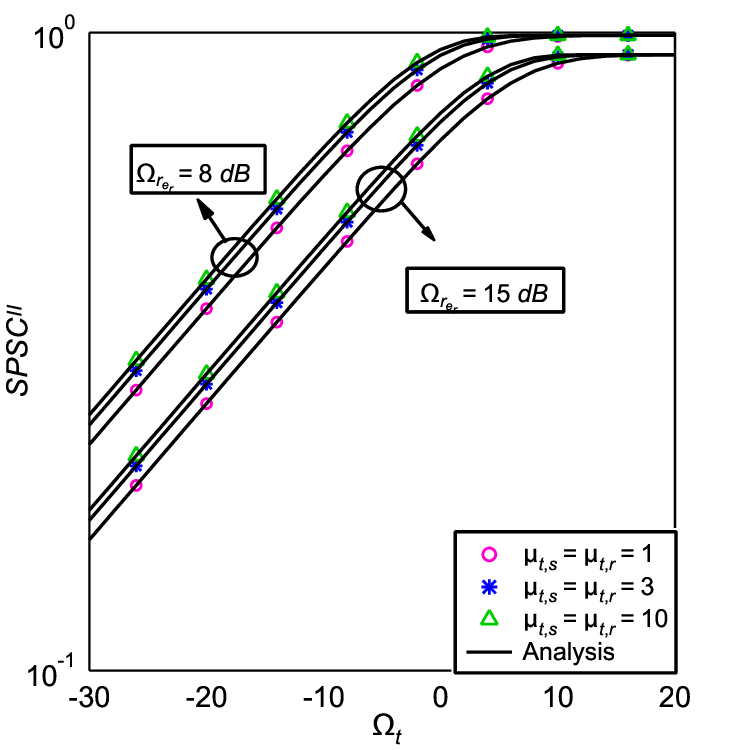}}
       \vspace{0mm}
   \caption{The $SPSC^{II}$ versus $\Omega_{t}$ for chosen values of $\mu_{t,s}$, $\mu_{t,r}$, and $\Omega_{r_{e_{r}}}$.}
   \label{p10}
\end{figure}
This is expected because higher $\mu$ introduce greater variability and unpredictability in the wireless channel. This increased randomness makes it harder for eavesdroppers to reliably capture and decode the transmitted signal, thus enhancing the secrecy of the communication. It is also observed in Fig. \ref{p10} that increasing $\Omega_{{r_{e_{r}}}}$ results in a degradation of security performance, as the eavesdropper benefits from higher signal quality in the $\mathcal{R}-\mathcal{I}_{{\mathcal{E}2}}-\mathcal{E}2$ link.

\begin{figure}[!ht]
\vspace{0mm}
\centerline{\includegraphics[width=0.35\textwidth,angle =0]{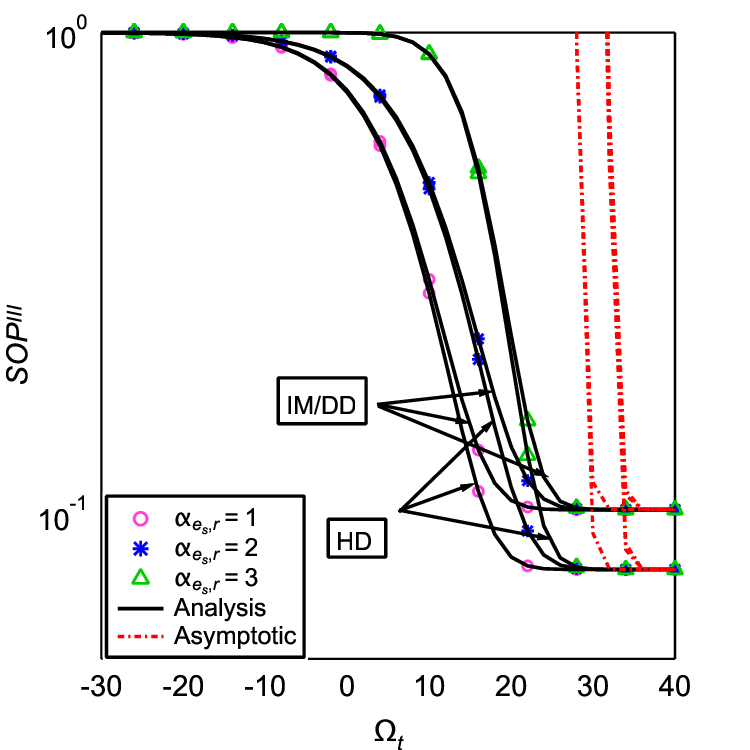}}
       \vspace{0mm}
   \caption{The $SOP^{III}$ versus $\Omega_{t}$ for chosen values of $\alpha_{e_{s},r}$ and $r_{d}$.}
   \label{p14}
\end{figure}
\begin{figure}[!ht]
\vspace{0mm}
\centerline{\includegraphics[width=0.35\textwidth,angle =0]{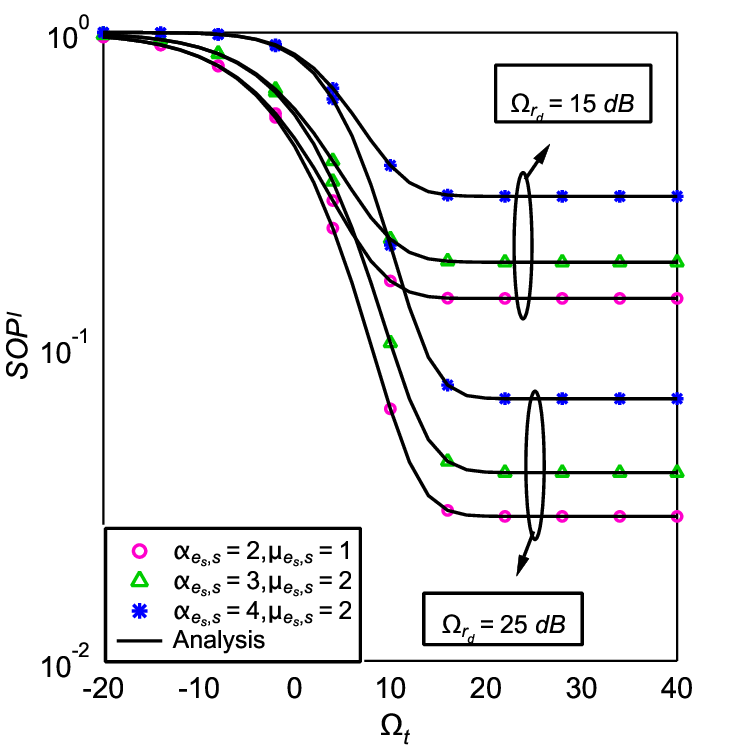}}
       \vspace{0mm}
   \caption{The $SOP^{I}$ versus $\Omega_{t}$ for chosen values of $\alpha_{e_{s},s}$, $\mu_{e_{s},s}$  and $\Omega_{r_{d}}$.}
   \label{p15}
\end{figure}
In contrast, a different outcome arises when examining the fading parameters (i.e., $\alpha_{e_{s},s}$, $\alpha_{e_{s},r}$, $\mu_{e_{s},s}$, and $\mu_{e_{s},r}$) of the $\mathcal{S}-\mathcal{I}_{{\mathcal{E}1}}-\mathcal{E}1$ link. For example, plotting $SOP^{III}$ against $\Omega_{t}$ in Fig. \ref{p14} indicates that increasing $\alpha_{e_{s},r}$ from $1$ to $3$ results in a rise in $SOP^{III}$ from $0.29$ to $0.63$ at $\Omega_{t}=6$dB. This occurs because higher $\alpha_{e_s}$ may lead to a more robust received signal at the eavesdropper's location, facilitating easier interception of the communication.
\begin{figure}[!ht]
\vspace{0mm}
\centerline{\includegraphics[width=0.35\textwidth,angle =0]{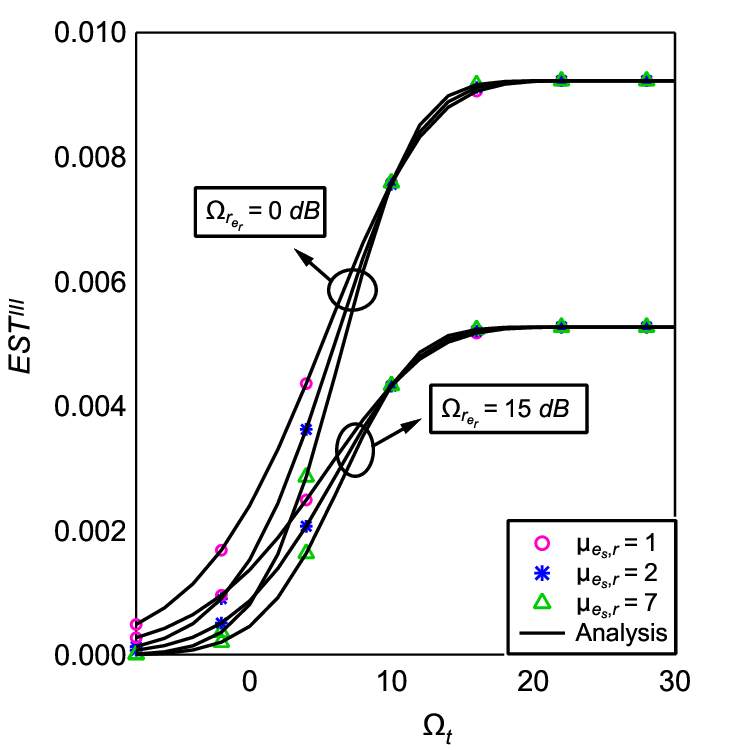}}
       \vspace{0mm}
   \caption{The $EST^{III}$ versus $\Omega_{t}$ for chosen values of $\mu_{e_{s},r}$ and $\Omega_{r_{e_{r}}}$.}
   \label{p13}
\end{figure}
Meanwhile, Fig. \ref{p15} illustrates the change in $SOP^{I}$ with $\Omega_{t}$ and demonstrates an increase in $SOP^{I}$ as $\mu_{e_{s},s}$ rises from $1$ to $2$. This outcome is anticipated as increased $\mu$ in the eavesdropper link could reduce the security margin. With reduced severity of fades and fluctuations in the eavesdropper's channel, the system becomes more susceptible to eavesdropping attempts, thereby compromising overall security. Similar results are also observed in Fig. \ref{p13}, albeit for the $\mathcal{I}_{{\mathcal{E}1}}-\mathcal{E}1$ link. Therefore, in Figure \ref{p15}, it is noticed that  
a higher value of $\Omega_{{r_{d}}}$ leads to superior security performance since the increased $\Omega_{{r_{d}}}$ strengths the $\mathcal{R}-\mathcal{I}_{{\mathcal{D}}}-\mathcal{D}$ link, hence reduces the eavesdroppers capability to intercept information.
\subsection{Impact of UOWC parameters}UOWC encounters challenges stemming from turbulence and pointing errors inherent in the underwater environment. To assess their effects on secrecy performance within the proposed model, Figs. 11-13 are showcased for demonstration.
\begin{figure}[!ht]
\vspace{0mm}
\centerline{\includegraphics[width=0.35\textwidth,angle =0]{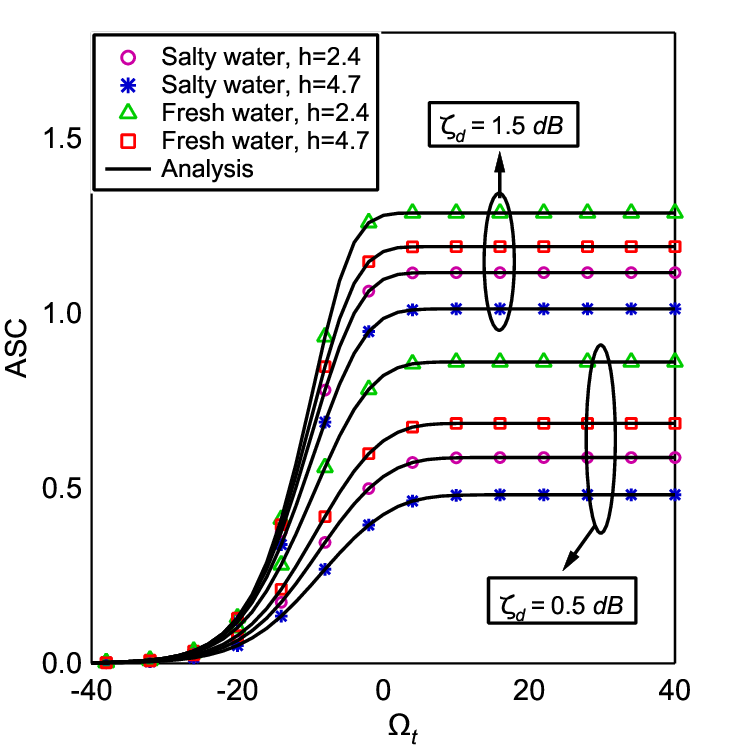}}
       \vspace{0mm}
   \caption{The ASC versus $\Omega_{t}$ for chosen values of h and $\zeta_{d}$.}
   \label{sf1}
\end{figure}
\begin{figure}[!ht]
\vspace{0mm}
\centerline{\includegraphics[width=0.35\textwidth,angle =0]{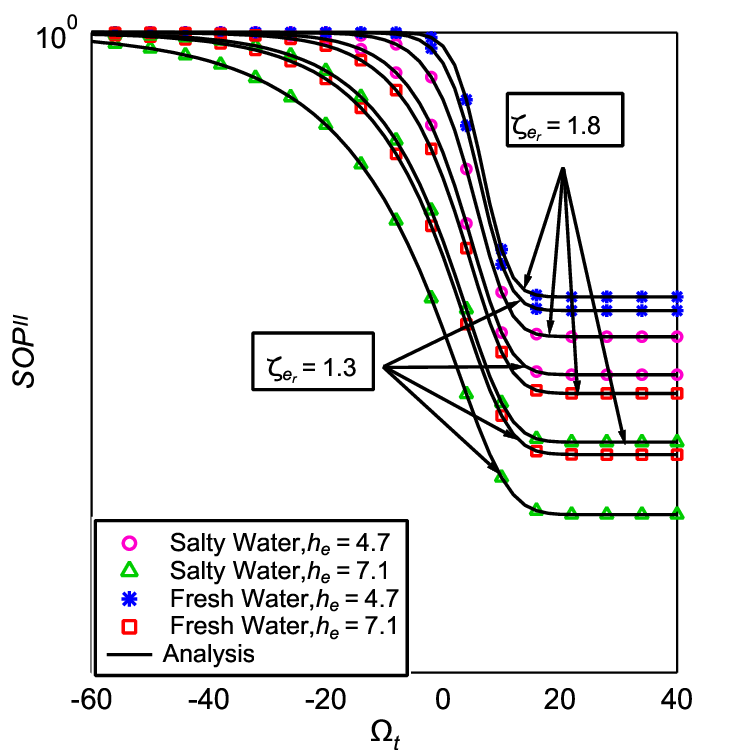}}
       \vspace{0mm}
   \caption{The $SOP^{II}$ versus $\Omega_{t}$ for chosen values of $h_{e}$ and $\zeta_{e_{r}}$.}
   \label{sf2}
\end{figure}
In UOWC systems under uniform thermal conditions, UWT is influenced by variations in both air bubble levels and water salinity. To assess these impacts, Figs. \ref{sf1}-\ref{sf2} demonstrate a graphical representation in terms of ASC and SOP secrecy metrics. The figures clearly illustrate that water salinity significantly affects the secrecy performance, with salty water posing a greater threat compared to fresh water. This outcome is anticipated due to the inherent characteristics of the two water types. Fresh water typically demonstrates lower attenuation of optical signals. On the other hand, salty water, characterized by a higher refractive index, is prone to increased light scattering, resulting in signal degradation and a reduced transmission range.

Pointing error can significantly affect the secrecy performance of optical communication systems which is shown in Figs. \ref{sf1}-\ref{sf2}. It is observed in Fig. \ref{sf1} that ASC performance becomes better when the value of $\zeta_d$ increases from $0.5$ to $1.5$. This is because increasing the value of $\zeta_d$ means the reduction of pointing error severity due to the $\mathcal{R}-\mathcal{I}_{{\mathcal{D}}}-\mathcal{D}$ link, hence, increasing pointing error parameters can enhance the security of the optical link by reducing the risk of unintentional signal leakage and interception. 
\begin{figure}[!ht]
\vspace{0mm}
\centerline{\includegraphics[width=0.35\textwidth,angle =0]{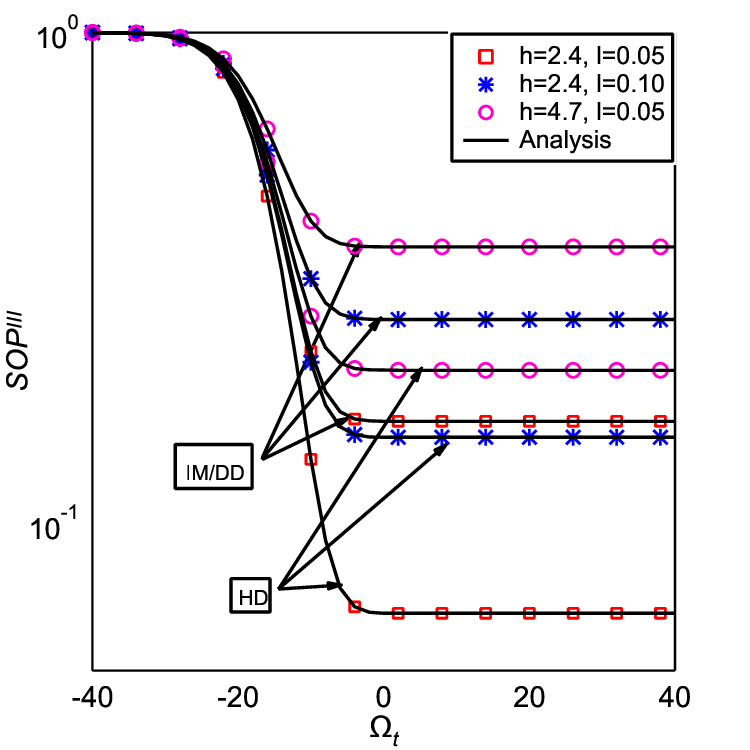}}
       \vspace{0mm}
   \caption{The $SOP^{III}$ versus $\Omega_{t}$ for chosen values of h and l.}
   \label{p8}
\end{figure}
Improved alignment ensures that the optical beams are directed precisely toward the intended receiver, minimizing the chances of unauthorized access to the transmitted data. This helps preserve the confidentiality and integrity of the communication link, particularly in sensitive or secure applications. Conversely, Figure \ref{sf2} delves into the impact of pointing error severity attributed to the $\mathcal{R}-\mathcal{I}_{{\mathcal{E}2}}-\mathcal{E}2$ link, revealing a degradation in performance with the increase of $\zeta_{e_{r}}$. This degradation occurs because a more accurately aligned eavesdropper can effectively decode the transmitted signal intended for the legitimate receiver. Consequently, the eavesdropper's capability to extract information from the transmission intensifies, thereby diminishing the level of secrecy achieved in the communication system.

 In Fig. \ref{p8}, a significant degradation in the SNR is evident with the increase in UWT parameters (i.e., air bubbles level and temperature gradient), which substantially impacts the SOP performance. This phenomenon is anticipated due to the inherent characteristics of turbulence, which induce scattering and absorption of light in the aquatic medium, consequently increasing signal attenuation. Therefore, this attenuation diminishes the optical signal strength during propagation, potentially restricting the communication link range and reliability. Consequently, the proposed model experiences a notable decrease in system performance under such conditions.

\begin{figure}[!ht]
\vspace{0mm}
\centerline{\includegraphics[width=0.35\textwidth,angle =0]{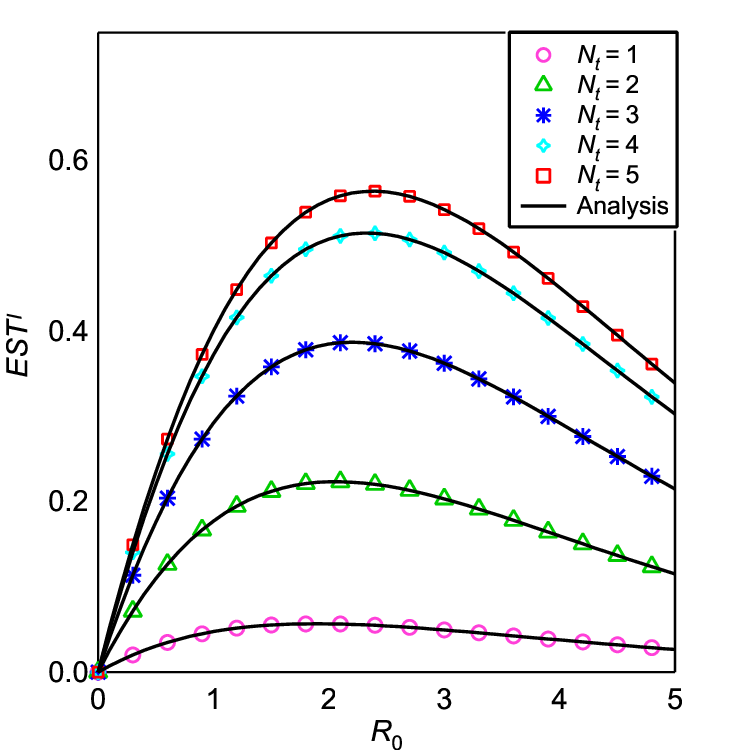}}
       \vspace{0mm}
   \caption{The $EST^{I}$ versus $R_{0}$ for chosen values of $N_{t}$.}
   \label{TR}
\end{figure}

Furthermore, Fig. \ref{p8} presents an extensive analysis comparing receiver detection methods, underscoring the superior effectiveness of HD over the IM/DD approach. The HD technique excels due to its ability to enhance receiver sensitivity and SNR by incorporating a local oscillator signal in the detection process. This integration enables improved detection of weak optical signals and mitigates noise interference. Additionally, the inherent coherence of HD preserves both the amplitude and phase information of the optical signal, contributing significantly to the performance and capacity enhancement of optical communication systems. Notably, analogous results are evident in the results depicted in Fig. \ref{p14}, further validating the advantages offered by HD in optical communication scenarios.
\subsection{Comparing with the Existing Works}

\begin{figure}[!ht]
\vspace{0mm}
\centerline{\includegraphics[width=0.35\textwidth,angle =0]{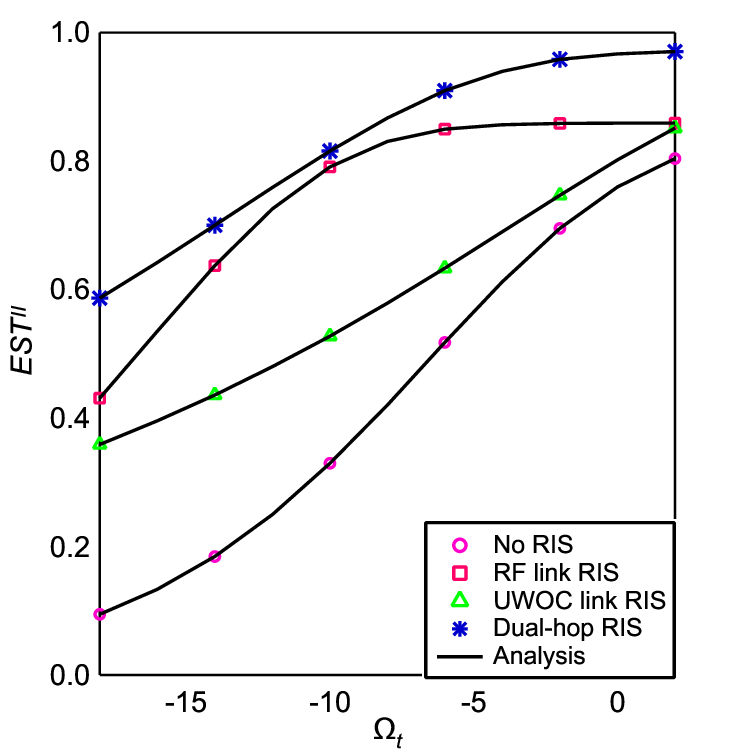}}
       \vspace{0mm}
   \caption{The $EST^{II}$ versus $\Omega_{t}$.}
   \label{e1}
\end{figure}

In Fig. \ref{TR}, the impact of target secrecy rate ($R_0$) on EST performance is observed under various $N_t$. It is evident that at lower $R_0$ values, the curve exhibits a sharp rise in EST with incremental increases in $R_0$. This notable increase is attributed to the fact that initially, achieving higher secrecy rates demands only minor adjustments or enhancements to the system. However, beyond a threshold of $3$ dB, the curve begins to display diminishing returns. At exceptionally high $R_0$ levels, the additional improvements in EST become marginal in comparison to the considerable resources and effort required to attain them. This phenomenon underscores the point at which further enhancements in secrecy rate may not be justifiable due to resource limitations \cite{new12}.

Fig. \ref{e1} illustrates the impact of RIS on system performance due to the RF-UOWC mixed model. The graphical representation reveals that integrating RIS leads to enhanced system performance, as evidenced by a notable increase in $EST^{II}$ (Scenario-$II$) when both hops leverage RIS technology, compared to scenarios where only one hop incorporates RIS or where RIS is absent entirely. This is expected because RIS optimizes system performance by facilitating signal reflection optimization, improving channel conditions, and mitigating interference, thereby increasing reliability and efficiency.
\begin{figure}[!ht]
\vspace{0mm}
\centerline{\includegraphics[width=0.35\textwidth,angle =0]{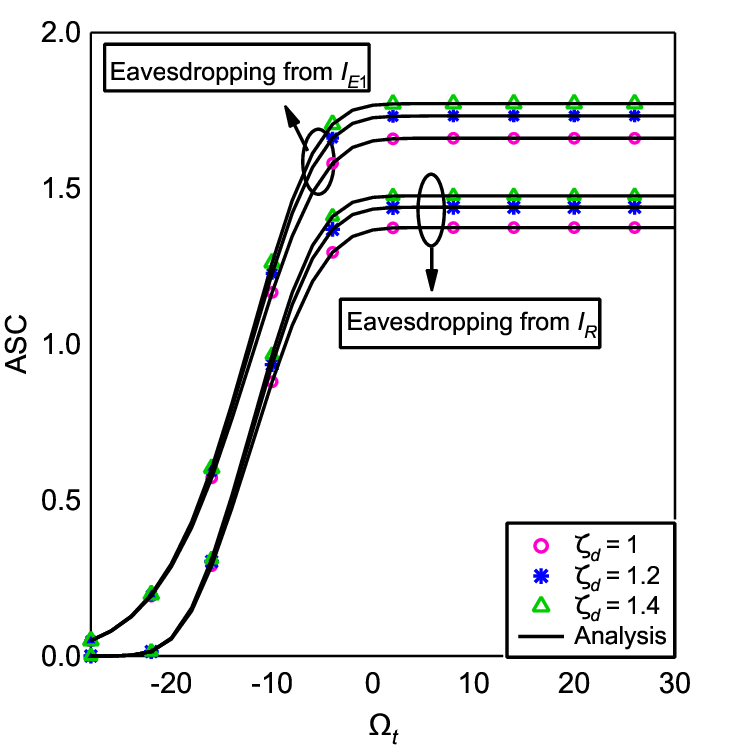}}
       \vspace{0mm}
   \caption{The ASC versus $\Omega_{t}$ for chosen values of $\zeta_{d}$}
   \label{e2}
\end{figure}
Furthermore, it is observed that integrating RIS into the RF link yields superior system performance compared to the UOWC link. This disparity arises from the fact that RF channels often face substantial challenges such as path loss, multipath fading, and interference, which can significantly degrade received signal quality and restrict coverage range. Deploying RIS in the RF channel enables dynamic manipulation of the propagation environment to surmount these obstacles by amplifying signal strength, compensating for path loss, and mitigating interference.
Despite the prevailing focus of current RF-Optical mixed models on utilizing RIS in a single hop only \cite{art16, ahmed2023enhancing, new9, art21, new10}, our proposed model demonstrates superior secrecy performance by incorporating RIS in both hops.

\begin{figure}[!ht]
\vspace{0mm}
\centerline{\includegraphics[width=0.35\textwidth,angle =0]{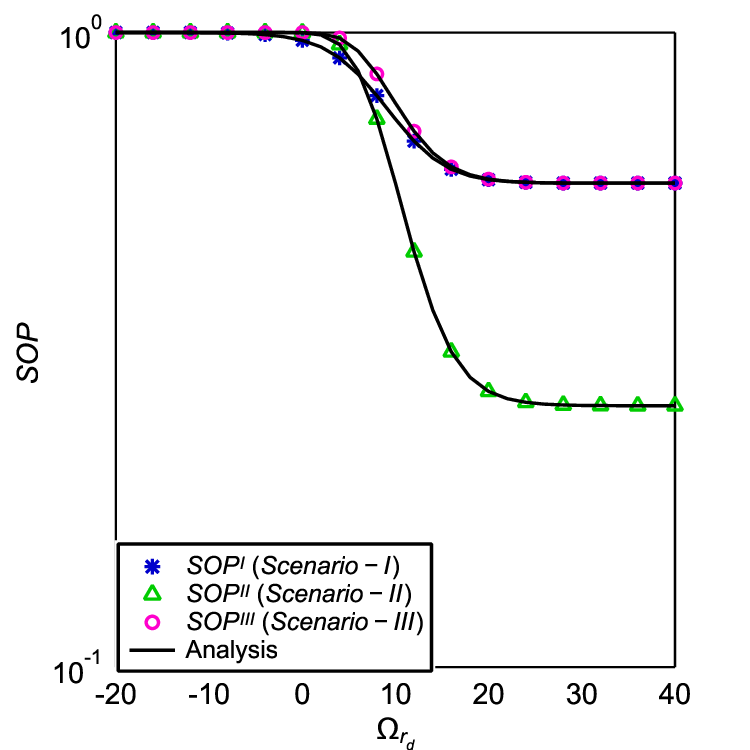}}
       \vspace{0mm}
   \caption{The SOP versus $\Omega_{r_{d}}$ for different eavesdropping scenarios.}
   \label{e3}
\end{figure}

Fig. \ref{e2} exhibits a comparative assessment of potential eavesdropping vulnerabilities, illustrating the relationship between ASC and $\Omega_{t}$ under the scenario where the eavesdropper attempts to compromise confidential information from both the main channel RIS (i.e., $\mathcal{I}_{\mathcal{R}}$) and the eavesdropper's own RIS (i.e., $\mathcal{I}_{\mathcal{E}1}$). The results demonstrate that ASC is significantly lower when the eavesdropper leverages $\mathcal{I}_{\mathcal{R}}$ for eavesdropping compared to using a separate RIS, $\mathcal{I}_{\mathcal{E}1}$, translating to a heightened risk of unauthorized access and exploitation of sensitive information. This observation is attributed to the eavesdropper's ability to directly access the reflections from $\mathcal{I}_{\mathcal{R}}$, leading to a substantially stronger signal compared to a separate RIS configuration. Furthermore, the influence of pointing errors in the UOWC link (i.e., $\zeta_{d}$) on ASC performance is examined in the figure. It is evident from the graph that as $\zeta_{d}$ values increase, the UOWC link exhibits enhanced secrecy performance by decreasing the risk of eavesdropping. This improvement mitigates pointing errors and facilitates a more accurate alignment between the transmitted optical beam and the receiver. Previous studies typically considered a scenario where eavesdroppers could only intercept information through the main channel RIS \cite{new11,art16,ahmed2023enhancing,new9}. In this work, we propose a more generalized model where eavesdroppers can exploit information using different RIS. This model offers greater versatility, as it can be reduced to the existing model under specific conditions, as describen in Remark 4.

Fig. \ref{e3} illustrates a comparative analysis of various eavesdropping scenarios, depicting the SOP against $\Omega_{r_{d}}$. The results underscore a notably increased severity in Scenario-$III$, where both eavesdroppers possess the capability to intercept information from both RF and UWOC links, simultaneously. Conversely, Scenario-$I$ exhibits greater vulnerability compared to Scenario-$II$. This discrepancy arises from the inherent characteristics of RF and optical signals: RF signals, being broadcasted, are accessible from a wider range of locations, while optical signals necessitate a direct line of sight between the transmitter and receiver. This inherent limitation makes intercepting optical transmissions more challenging for eavesdroppers.


\section{CONCLUSION}
\label{con}
In this paper, the secrecy performance of a RIS-assisted mixed RF-UOWC system was studied under both RF and UOW eavesdropping attack. To assess the influence of system parameter, novel closed-form expressions of ASC, SOP, SPSC, and EST performance were derived. The analytical expressions was validated through Monte Carlo simulations. Further, asymptotic analyses yielded insights into system behavior under high-SNR conditions. Numerical results demonstrated that increasing RIS elements, mitigating fading and pointing errors in the main channel, and employing the HD detection at the receiver for UOWC link significantly enhance secrecy performance. Conversely, factors like strong eavesdropper channels and severe underwater turbulence (air bubbles, temperature gradients) negatively impact secrecy. Notably, simultaneous eavesdropping on both RF and UOWC links poses the greatest security challenge.

\bibliographystyle{IEEEtran}
\bibliography{IEEEabrv,main}

\end{document}